\documentclass[reqno, 11pt]{amsart}
\pdfoutput=1

\usepackage{amsmath,amssymb,amsthm,amsfonts}
\usepackage{enumitem}
\usepackage{mathtools}
\usepackage[normalem]{ulem}
\usepackage{mathrsfs}
\usepackage{hyperref}

\newtheorem{theorem}{Theorem}[section]
\newtheorem{lemma}[theorem]{Lemma}
\newtheorem{prop}[theorem]{Proposition}
\newtheorem{cor}[theorem]{Corollary}

\theoremstyle{definition}

\newcommand{\e}{\mathrm{e}}
\newcommand{\N}{\mathbb{N}}
\newcommand{\Z}{\mathbb{Z}}
\newcommand{\R}{\mathbb{R}}
\newcommand{\C}{\mathbb{C}}
\renewcommand{\P}{\mathbb{P}}
\newcommand{\X}{\mathbb X}
\newcommand{\V}{\mathbb V}
\newcommand{\U}{\mathbb U}
\newcommand{\A}{\mathbb A}

\newcommand{\be}{\begin{equation}}
\newcommand{\ee}{\end{equation}}

\renewcommand{\Re}{\mathrm{Re}\,}
\newcommand{\Conv}{\mathop{\scalebox{2}{\raisebox{-0.275ex}{$\ast$}}}}

\def\1{{\mathchoice {1\mskip-4mu\mathrm l}      
{1\mskip-4mu\mathrm l}
{1\mskip-4.5mu\mathrm l} {1\mskip-5mu\mathrm l}}}

\begin{document}

\title{A Multi-Body Dobrushin-Sokal Criterion -- Part II} 
\date{June 27, 2026}
\author[J. P.~Neumann]{Jan Philipp Neumann}

\begin{abstract}
	We prove a sufficient condition for the absolute convergence of Mayer cluster expansions of log-partition and correlation functions applicable to lattice gases with possibly complex-valued multi-body interactions. 
	Not only are several classical results subsumed but a partition scheme for spanning hypergraphs also makes our methods well-suited for treating stronger multi-body interactions, including higher-order hard-core repulsion in the context of hypergraph independence polynomials. 
	Furthermore, our approach is easily combined with the Gruber--Kunz condition to produce extended convergence results for the polymer expansion of lattice gases, rivalling those obtained not too long ago by Nguyen and Fern{\'a}ndez (2024).
	
	\bigskip 
		
    \noindent \emph{Mathematics Subject Classification}: 82B20 (05C65, 05C70).
    
	\medskip     
    
    \noindent \emph{Keywords}: Classical lattice gas; multi-body interactions; grand partition function; cluster expansions; Kirkwood--Salsburg hierarchy; polymer expansion.
\end{abstract}

\maketitle

\section{Introduction}

The present article deals with new convergence criteria for the cluster and polymer expansions of classical lattice gases with multi-body interactions. 
Although essentially self-contained, it complements \cite{ref:neumann-part-i}, henceforth just ``Part~I'', where the author derives a sufficient condition for non-vanishing of partition functions. 
Results on the convergence of cluster expansions are only alluded to there. 
Here, they are stated and proven. 

For pairwise interactions, there is an extensive body of literature on the matter: 
the Koteck{\'y}--Preiss condition, extended by Ueltschi and Poghosyan, see \cite{ref:kotecky-preiss, ref:ueltschi, ref:poghosyan-ueltschi}, is concise, widely applicable and used to this day; 
a slightly superior and similarly concise criterion is that of Dobrushin and Sokal for repulsive lattice gases, see \cite{ref:dobrushin-semi-invariants, ref:dobrushin-saint-flour, ref:sokal, ref:scott-sokal}; 
the even better, though a little more involved, Fern{\'a}ndez--Procacci condition is again amenable to wider generalisation, see \cite{ref:fernandez-procacci, ref:bissacot-fernandez-procacci, ref:faris, ref:jansen-cluster, ref:jansen-kolesnikov}. 
Particularly in the last one of these references, Jansen and Kolesnikov thoroughly described how to leverage the Kirkwood--Salsburg hierarchy for correlation functions to obtain a systematic approach to convergence criteria, including all of the above.

This systematic approach remains perfectly valid in the case of multi-body, i.e., higher-order, interactions but the previously trivial combinatorics of the Kirkwood--Salsburg integral kernel become considerably more complicated. 
These complications can be partially sidestepped, provided a two-particle potential amenable to, say, the generalised Koteck{\'y}--Preiss condition is extended by a stable higher-order interaction of finite range, cf.\ \cite{ref:moraal, ref:skrypnyk}, but, compared to the purely pairwise case, the intermediate estimates seem somewhat crude and we are unwilling to assume finite range outright. 
Instead, we tackle the diagrammatic interpretation of the kernel and use a novel partition scheme for spanning hypergraphs (Lemma~\ref{lem:spanning-hypergraph-partition-scheme}) to derive our first main result, Theorem~\ref{thm:dobrushin-sokal}, an extension of the Dobrushin--Sokal criterion to classical lattice gases with site-wise hard-core self-repulsion and otherwise complex multi-body interactions. 
As in the conditions listed above for pairwise interactions, only one inequality per site is required, involving only a single auxiliary function, called $\alpha$ here, on the underlying lattice in addition to the innate model-defining parameters. 
A variant of this combination of the Kirkwood--Salsburg hierarchy with our partition scheme for coverings is included in Part~I as an alternative proof method. 
But only here do we extend its application to derive convergence of cluster expansions.

Our main result comes on the heels of recent progress regarding zero-free activity polydiscs for hypergraph independence polynomials, i.e., partition functions of lattice gases with pure multi-body hard-core repulsions. 
More precisely, our original derivation of Theorem~\ref{thm:dobrushin-sokal}, cf.\ Part~I, consisted in generalising approaches by Galvin et al.\ in \cite{ref:galvin-et-al} and Bencs and Buys in \cite{ref:bencs-buys}. 
The latter two authors were able to reproduce the same optimal uniform radii in terms of global degree bounds that had previously only been established for standard graphs, i.e., for pairwise hard-core interactions, cf., e.g., \cite{ref:scott-sokal}. 
Their intricate way of isolating the interaction terms from each other is functionally replaced here by inserting our aforementioned partition scheme into the Kirkwood--Salsburg hierarchy. 
In turn, this replacement was spurred on by the realisation that, if we choose the constant auxiliary function $\alpha = 1$ in our condition, we obtain a direct improvement of Gallavotti and Miracle-Sol{\'e}'s condition in \cite{ref:gallavotti-miracle-sole}, derived explicitly by bounding the Kirkwood--Salsburg kernel. 
The latter improvement, which extends to the variants in \cite{ref:gallavotti-miracle-sole-robinson, ref:ruelle}, can, in fact, be derived by a banal observation without any need for our partition scheme. 
The same is true of our criterion whenever $\alpha \geq 1$ but the partition scheme is so far necessary to bridge the gap to Galvin et al.'s result as well as that of Bencs and Buys, where $\alpha < 1$, cf.\ Section~5 of Part~I.

There are more results on cluster expansions with multi-body interactions out there. 
The interested reader may have a look at the introductions of \cite{ref:jansen-tsagkarogiannis, ref:galvin-et-al, ref:nguyen-fernandez} and the references listed therein. 
Here, however, we avoid overly technical a priori assumptions and still derive multi-body variants of the Dobrushin--Sokal and Koteck{\'y}--Preiss conditions, at least in the lattice gas setting. 
There is also the recent preprint \cite{ref:helmuth-pappik-perkins} on analyticity in continuum systems with repulsive multi-body potentials of finite range. 
Although it does not consider cluster expansions, it would be interesting to see if this result of Helmuth, Pappik and Perkins's can be extended to a more general class of systems.

There is yet another aspect to cluster expansions for lattice gases, also merely alluded to in Part~I. 
As described by Gruber and Kunz in \cite{ref:gruber-kunz}, the partition function of an essentially arbitrary lattice spin system equals that of another lattice gas whose sites/particles are non-empty finite subsets, dubbed ``polymers'', of the initial lattice, see also \cite{ref:procacci-scoppola, ref:nguyen-fernandez}. 
More importantly, the polymers in this reformulation only interact via a pairwise hard-core interaction prohibiting overlap so all of the classical cluster expansion results apply. 
Given the special structure of polymers, we, as well as others, still opt for the Gruber--Kunz convergence condition, itself a corollary to the Fern{\'a}ndez--Procacci criterion or a polymer-specific exact criterion, see \cite{ref:fernandez-procacci, ref:bissacot-fernandez-procacci} or \cite{ref:jansen-kolesnikov}, respectively. 
The caveat in this simplification is the fact that the combinatorial complexities of the previous cluster expansions are part of computing and bounding the individual polymer activities. 

Nevertheless, it contains a desirable (from a physical perspective) transition from pre-polymer activity $z$ to $\frac{z}{1+z}$ and our first approach can be adapted to yield a convergence condition, see Theorem~\ref{thm:dobrushin-sokal-polymer}, which extends to parameter ranges beyond where Theorem~\ref{thm:dobrushin-sokal} applies. 
In particular, we obtain notable improvements over the conditions in Procacci and Scoppola's \cite{ref:procacci-scoppola} as well as Gallavotti, Miracle-Sol{\'e} and Robinson's \cite{ref:gallavotti-miracle-sole-robinson}, cf.\ also \cite{ref:ruelle}, in the form of Corollary~\ref{cor:gallavotti-miracle-sole-robinson-refined} here. 
For a comparison with a more recent result, we consider that of Nguyen and Fern{\'a}ndez in \cite{ref:nguyen-fernandez}. 
While their approach applies to far more general lattice systems, our method has the advantage of covering stronger interactions, including hard-core ones.

The paper is structured as follows: 
Section~\ref{sec:results} introduces the basic notation, states and discusses the first main result, Theorem~\ref{thm:dobrushin-sokal}, and then does the same for the second main result, Theorem~\ref{thm:dobrushin-sokal-polymer}, on the polymer expansion; 
Section~\ref{sec:prelims-convolution-poisson} builds the framework for efficiently manipulating the relevant power series in what is sometimes called ``Ruelle's algebraic method''; 
this allows us to systematically extend our lattice gas notation in Section~\ref{sec:prelims-conditional-root-KS}, which also ends with our partition scheme for spanning hypergraphs in the form of Lemma~\ref{lem:spanning-hypergraph-partition-scheme}; 
the final Section~\ref{sec:proof-main} collects the proofs of our main results.

\section{Setup and main results} \label{sec:results}

\subsection{Configurations on a lattice}

Our \emph{lattice} $\X$ is a fixed finite or countably infinite set of \emph{sites} for particles to occupy. 
Ignoring assignments of infinitely many particles, a \emph{configuration} is a \emph{multi-index} in
\[
	\mathbf N : = \mathbf N(\X) : = \{\mathbf n = (n_x)_{x \in \X} \in \N_0^\X \mid |\mathbf n| < \infty\},
\]
where $|\mathbf n| : = \mathbf n[\X]$ with $\mathbf n[S] : = \sum_{x \in S} n_x$ for arbitrary $S \subset \X$. 
In view of the evident identification with finite counting measures, we use the \emph{Dirac measure} $\delta_x$ in $x \in \X$ to denote the multi-index $\1_{\{x\}}$. 
It is also often convenient to enumerate or label particles. 
More specifically, for an arbitrary finite index set $J$ of cardinality $|J| \in \N_0$, we use the symmetrisation map
\[
	\X^J = \{\mathbf x = (x_j)_{j \in J} \mid \forall j \in J : x_j \in \X\} \to \{\mathbf n \in \mathbf N \mid |\mathbf n| = |J|\}, \quad \mathbf x \mapsto \sum_{j \in J} \delta_{x_j},
\]
to identify a $|J|$-particle configuration with any \emph{tuple} in its preimage. 
Furthermore, configurations with at most one particle per site are called \emph{simple} here and are identified with finite \emph{subsets} via the bijection
\[
	\mathbf F : = \mathbf F(\X) : = \{X \Subset \X\} \to \mathbf N \cap \{0,1\}^\X, \quad X \mapsto \1_X = \sum_{x \in X} \delta_x.
\]
Another convenient identification consists in treating the obvious bijection
\[
	\mathbf N(\Lambda) \to \{\mathbf n \in \mathbf N \mid \mathbf n[\X \setminus \Lambda] = 0\}
\]
for arbitrary $\Lambda \subset \X$ as an inclusion of sets, mirroring the a priori inclusions $\Lambda^J \subset \X^J$ and $\mathbf F(\Lambda) \subset \mathbf F$.

\subsection{Simple interaction, Boltzmann factor and Ursell function}

We now fix an \emph{interaction potential} $V : \mathbf N \to \C \cup \{+\infty\}$. 
Both $V$ and the \emph{Mayer function} $f : = \e^{-V} - 1$, where $\e^{-\infty} : = 0$, are also defined on finite tuples over and subsets of $\X$ via the above identifications. 
The \emph{Boltzmann factor} $\kappa : \mathbf N \to \C$ can then be defined by setting
\[
	\kappa(\mathbf x) : = \e^{\sum_{S \subset J : S \neq \varnothing} V(\mathbf x_S)} = \prod_{S \subset J : S \neq \varnothing} (1 + f(\mathbf x_S)) = \sum_{\mathfrak h \in \mathfrak H_J} \prod_{e \in \mathfrak h} f(\mathbf x_e)
\]
for every tuple $\mathbf x \in \X^J$ with finite index set $J$, where $\mathbf x_S : = (x_j)_{j \in S}$ is the restriction/projection of $\mathbf x$ to an index subset $S \subset J$ while
\[
	\mathfrak H_\V : = \{\mathfrak h \subset \{e \Subset \V \mid e \neq \varnothing\}\}
\]
denotes the set of \emph{hypergraphs} (sets of \emph{edges}) on an arbitrary \emph{vertex} set $\V$. 

The characterisation of $\kappa$ in terms of hypergraph weights provides a direct and intuitive approach to characterising the \emph{Ursell function}, which is the unique $\varphi : \mathbf N \to \C$ satisfying $\varphi(0) = 0$ and
\be \label{eq:kappa-sum-prod-phi}
	\kappa(\mathbf x) = \sum_{\mathfrak p \in \mathfrak P_J} \prod_{P \in \mathfrak p} \varphi(\mathbf x_P)
\ee
for every tuple $\mathbf x \in \X^J$ with finite index set $J$, where we denote by
\[
	\mathfrak P_U : = \{\mathfrak p \subset \{P \subset U \mid P \neq \varnothing\} \mid U = \textstyle \bigsqcup_{P \in \mathfrak p} P\}
\]
the set of unordered \emph{partitions} of an arbitrary set $U$ into non-empty subsets. 
Clearly, for every set $\V$, we have a cluster decomposition
\[
	\mathfrak H_\V = \bigsqcup_{\mathfrak p \in \mathfrak P_\V} \{\mathfrak h = \textstyle \bigsqcup_{P \in \mathfrak p} \mathfrak h_P \mid \forall P \in \mathfrak p : \mathfrak h_P \in \mathfrak C_P\},
\]
denoting by $\mathfrak C_\V \subset \mathfrak H_\V$ the set of \emph{cluster hypergraphs} on $\V$, i.e., for $\V \neq \varnothing$, $\mathfrak C_\V$ coincides with the set of all $\mathfrak h \in \mathfrak H_\V$ such that $\V$ is \emph{connected} with respect to $\mathfrak h$ in the standard sense whereas $\mathfrak C_\varnothing : = \varnothing$ since connected components are always non-empty. 
One easily verifies that the characterisation of $\varphi$ via \eqref{eq:kappa-sum-prod-phi} necessitates
\[
	\varphi(\mathbf x) = \sum_{\mathfrak h \in \mathfrak C_J} \prod_{e \in \mathfrak h} f(\mathbf x_e)
\]
for every tuple $\mathbf x \in \X^J$ with finite index set $J$.

Throughout the entire article, we assume $V$ to be \emph{simple} in the sense that $\kappa$ vanishes outside $\mathbf F$ or, equivalently, $\kappa(2 \delta_x) = \e^{- 2 V(x) - V(2 \delta_x)} = 0$ for all $x \in \X$. 
A lot of the theory presented here does not depend on this blanket assumption but our actual results rely on it.

\subsection{Partition function and Mayer series}

If $z = (z(x))_{x \in \X}$ is a collection of commuting abstract \emph{activity} variables, we can define the formal power series
\[
	Z(z) : = \sum_{\mathbf n \in \mathbf N} \frac{z^\mathbf n}{\mathbf n!} \kappa(\mathbf n) \quad \text{and} \quad M(z) : = \sum_{\mathbf n \in \mathbf N} \frac{z^\mathbf n}{\mathbf n!} \varphi(\mathbf n),
\]
respectively called grand \emph{partition function} and \emph{Mayer series}, without issue. 
In these multivariate series, we use the notation
\[
	z^\mathbf n : = \prod_{x \in \X} z(x)^{n_x} \quad \text{and} \quad \mathbf n! : = \prod_{x \in \X} n_x!
\]
for $\mathbf n \in \mathbf N$ (only finitely many factors differ from $1$) and the definition of $\varphi$ is precisely equivalent to the identity $Z(z) = \exp(M(z))$. 
It readily follows that, given $x \in \X$, the three series
\[
	Z(z)^{-1} \frac{\partial Z}{\partial z(x)} (z), \quad \frac{\partial M}{\partial z(x)} (z) \quad \text{and} \quad \sum_{\mathbf n \in \mathbf N} \frac{z^\mathbf n}{\mathbf n!} \varphi(\delta_x + \mathbf n) = : \rho(x, z)
\]
are formally identical. 
We call the last one the (reduced) one-point \emph{correlation} of $x$. 
Restriction of these series to some \emph{reference volume} $\Lambda \subset \X$ is implemented via
\[
	Z_\Lambda(z) : = Z(z \1_\Lambda), \quad M_\Lambda(z) : = M(z \1_\Lambda) \quad \text{and} \quad \rho_\Lambda(x, z) : = \rho(x, z \1_\Lambda).
\]
It is particularly noteworthy that our restricted correlation does not categorically vanish as soon as $x \in \X \setminus \Lambda$, in which case we refer to
\[
	\widehat z_\Lambda(x, z) :  = z(x) \rho_\Lambda(x, z)
\]
as an \emph{effective activity}, a notion taken from, e.g., \cite{ref:scott-sokal, ref:jansen-tsagkarogiannis, ref:jansen-hierarchical, ref:jansen-neumann}. 
This is a slight departure from the standard notation due to convenience here.

\subsection{The Dobrushin-Sokal criterion}

The simplicity of $V$ means that, for every $x \in \X$, we can write
\[
	Z(z) = \sum_{X \Subset \X} z^X \kappa(X) = \sum_{X \Subset \X \setminus \{x\}} z^X \kappa(X) + z(x) \sum_{Y \Subset \X \setminus \{x\}} z^Y \kappa(\{x\} \cup Y).
\]
The second identity is formally equivalent to the observation that
\[
	Z_{\{x\} \cup \Lambda}(z) = Z_\Lambda(z) (1 + \widehat z_\Lambda(x, z)),
\]
for every $x \in \X$ and $\Lambda \subset \X \setminus \{x\}$, which comes with the two related formulas
\[
	M_{\{x\} \cup \Lambda}(z) - M_\Lambda(z) = \sum_{m \in \N} \frac{(-1)^{m-1}}{m}\widehat z_\Lambda(x, z)^m = \log(1 + \widehat z_\Lambda(x, z))
\]
and
\[
	z(x) \rho_{\{x\} \cup \Lambda}(x, z) = \widehat z_\Lambda(x, z) \sum_{m \in \N_0} (-1)^m \widehat z_\Lambda(x, z)^m = \frac{\widehat z_\Lambda(x, z)}{1 + \widehat z_\Lambda(x, z)}.
\]
Note that the above logarithmic and geometric series have radius of convergence equal to 1, indicating why we aim for bounds of the form $|\widehat z_\Lambda(x, z)| < 1$ upon evaluation in $\C$ below.

Fix two functions $r : \X \to [0, 1)$ and $\alpha : \X \to \R_+ := [0, +\infty)$ related by the two equivalent pointwise identities $r = \frac{\alpha}{1 + \alpha}$ and $\alpha = \frac{r}{1 - r}$. 
We call $\Lambda \subset \X$ $r$-finite if the mutually equivalent conditions
\[
	r[\Lambda] := \sum_{x\in\Lambda} r(x) < \infty \quad \text{and} \quad (1 + r)^\Lambda := \prod_{x\in \Lambda} (1+r(x)) < \infty
\]
hold. 
These sums and products a priori exist in $\overline \R_+:=[0,+\infty]$ by monotone convergence. 
Analogous notation, particularly $\alpha^S = \prod_{x\in S} \alpha(x) \in \R_+$ for $S\Subset X$, is frequently used in the following. 
Note also that
\[
	(1 - r)^\Lambda = \frac{1}{(1 + \alpha)^\Lambda} > 0 \quad \text{and} \quad - \log(1 - r)[\Lambda] = \log(1 + \alpha)[\Lambda] < \infty,
\]
where $1/\infty = 0$, are each equivalent to $\Lambda$ being $r$-finite.

\begin{theorem} \label{thm:dobrushin-sokal}
	Let $z : \X \to \C$ and suppose that $V$ is a simple interaction potential with Mayer function $f = \e^{-V} - 1$ and Ursell function $\varphi$ satisfying
	\be \label{eq:dobrushin-sokal}
		|z(x)| \prod_{X \Subset \X : x \in X} \max\{|\e^{-V(X)}|, 1 + |f(X)| \alpha^S \mid \varnothing \neq S \subset X \setminus \{x\}\} \leq r(x)
	\ee
	for all $x \in \X$. 
	Then, for all $x \in \X$, we have
	\be \label{eq:zhat-leq-r}
		\sup_{\Lambda \subset \X \setminus \{x\}}|\widehat z_\Lambda(x, z)| \leq \sum_{\mathbf n \in \mathbf N(\X \setminus \{x\})} \frac{|z|^{\delta_x + \mathbf n}}{\mathbf n!} |\varphi(\delta_x + \mathbf n)| \leq r(x) < 1
	\ee
	and
	\[
		\sup_{\Lambda \subset \X} |z(x) \rho_\Lambda(x, z)| \leq \sum_{\mathbf n \in \mathbf N} \frac{|z|^{\delta_x + \mathbf n}}{\mathbf n!} |\varphi(\delta_x + \mathbf n)| \leq \alpha(x) < \infty.
	\]
\end{theorem}

We then use the effective activity bounds to control partition functions and Mayer series as indicated above. 

\begin{cor}	\label{cor:Z-M-r-bounds}
	Let $z : \X \to \C$ and suppose that $V$ is a simple interaction potential with Boltzmann factor $\kappa$ and Ursell function $\varphi$ satisfying \eqref{eq:zhat-leq-r} for every $x \in \X$. Then, for every $r$-finite $\Lambda \subset \X$, we have
	\[
		0 < (1 - r)^\Lambda \leq |Z_\Lambda(z)| \leq \sum_{X \Subset \Lambda} |z|^X |\kappa(X)| \leq (1 + r)^\Lambda < \infty,
	\]
	\[
		|M_\Lambda(z)| \leq \sum_{\mathbf n \in \mathbf N(\Lambda)} \frac{|z|^\mathbf n}{\mathbf n!} |\varphi(\mathbf n)| \leq - \log(1 - r)[\Lambda] < \infty
	\]
	and, in particular, $Z_\Lambda(z) = \exp(M_\Lambda(z)) \in \C \setminus \{0\}$.
\end{cor}

Together, Theorem~\ref{thm:dobrushin-sokal} and Corollary~\ref{cor:Z-M-r-bounds} constitute a stronger version of the main result of Part~I. 
While the condition \eqref{eq:dobrushin-sokal} is the same, the conclusions regarding the absolute convergence of cluster expansions, in particular that of the Mayer series, are new. 
The cluster expansion setting is also what easily allows the extension to arbitrary reference volumes $\Lambda \subset \X$ in Theorem~\ref{thm:dobrushin-sokal} and at least $r$-finite ones in Corollary~\ref{cor:Z-M-r-bounds}.

\subsection{A few comments on our main result}

The initial term $|\e^{-V(X)}|$ inside each maximum in \eqref{eq:dobrushin-sokal} corresponds to $S = \varnothing$. 
For $X = \{x\}$, it is the only term and is traditionally absorbed into $z(x)$. 
For $X \supsetneq \{x\}$, the term is obviously irrelevant to the maximum whenever the potential $V(X)$ is repulsive in the sense that $|e^{-V(X)}| \leq 1$. 
Our titular naming choice of the criterion \eqref{eq:dobrushin-sokal} comes from the fact that, for repulsive pair potentials, i.e., $\Re V \geq 0$ and $V(X) = V(X) \1_{\{|X| = 2\}}$ for all $X \Subset \X$, it turns into
\[
	|z(x)| \prod_{y \in \X \setminus \{x\}} (1 + |f(\delta_x + \delta_y)| \alpha(y)) \leq r(x) = \frac{\alpha(x)}{1 + \alpha(x)},
\]
which is just another way of writing Sokal's condition from \cite{ref:sokal, ref:scott-sokal}. 
Dobrushin had previously established this condition for hard-core pair interactions, where $|f(\delta_x+\delta_y)| \in \{0,1\}$ for all $x, y \in \X$, see \cite{ref:dobrushin-semi-invariants, ref:dobrushin-saint-flour}.
Observe also that, even in the multi-body case, \eqref{eq:dobrushin-sokal} just becomes
\[
	|z(x) \e^{-V(x)}| (1 + \alpha(x)) \prod_{X \Subset \X : \{x\} \subsetneq X} (1 + |f(X)| \alpha^{X \setminus \{x\}}) \leq \alpha(x)
\]
whenever $\alpha \geq 1$, i.e, $r \geq 1/2$, pointwise.

Dobrushin's condition is itself a refinement of the one Koteck{\'y} and Preiss introduced in \cite{ref:kotecky-preiss}, also dealing only with hard-core pair potentials. 
Later extensions to repulsive and even stable pair interactions were made by Ueltschi and then Poghosyan and Ueltschi in \cite{ref:ueltschi} and \cite{ref:poghosyan-ueltschi}, respectively. 
A priori stability assumptions are ubiquitous in equilibrium statistical mechanics, cf., e.g., \cite{ref:ruelle}, \cite{ref:gallavotti-miracle-sole, ref:gallavotti-miracle-sole-robinson}, \cite{ref:poghosyan-ueltschi, ref:procacci-yuhjtman, ref:jansen-cluster, ref:jansen-kolesnikov}. 
They are used to control non-repulsive parts of $V$. 
Here we note that \eqref{eq:dobrushin-sokal} implies the (local) stability condition
\[
	|z(x) \e^{-V(x)}| \e^{c(x)} = |z(x) \e^{-V(x)}| \prod_{X \Subset \X : \{x\} \subsetneq X} \max\{|\e^{-V(X)}|, 1\} \leq r(x)
\]
with $c(x) := \sum_{X \Subset \X : \{x\} \subsetneq X} \max\{-\Re V(X), 0\}$. 
In fact, we do not see room for direct improvement of our main result through intermediate stability bounds. 
However, for a close comparison with extensions of the Koteck{\'y}--Preiss condition, let $V(x) = 0$ for all $x \in \X$ and suppose that $\alpha = |z| \e^{c + a}$ with $a : \X\to\R_+$. 
In this case, \eqref{eq:dobrushin-sokal} holds whenever
\[
	|z(x)| e^{c(x) + a(x)} + \sum_{X \Subset \X : \{x\} \subsetneq X} |f(X)| \max_{S \subset X \setminus \{x\} : S \neq \varnothing}\{ (|z| \e^{c + a})^S\} \leq a(x)
\]
and, for pairwise interactions, this essentially reads like the condition from \cite{ref:poghosyan-ueltschi}, cf.\ also \cite{ref:jansen-kolesnikov}. 
The summand $|z(x)| e^{c(x)+a(x)}$ implicitly has the prefactor $|f(2\delta_x)| = 1$. 
Returning to the formulation in terms of $\alpha$, a well-known optimisation argument, omitted here but included in Part~I for the reader's convenience, yields the following.

\begin{cor}
	Let $z : \X \to \C$ and suppose that $V$ is a simple interaction potential with Mayer function $f = \e^{-V} - 1$ satisfying $V(x)=0$ as well as
	\[
		|z(x)| \leq \e^{-c(x)-1} C^{-1}
	\]
	for all $x \in \X$, where $C : = \sup\{1 + \sum_{X \Subset \X : \{y\} \subsetneq X} |f(X)| \mid y \in \X\} < \infty$. Then the conclusions of Theorem~\ref{thm:dobrushin-sokal} hold with constant $\alpha = C^{-1} \leq 1$.
\end{cor}

As far as we can tell, this result is new in the case of multi-body potentials. 
For pairwise interactions, it is essentially a classic, see, e.g., \cite{ref:ruelle}, and applies to more abstract settings like in \cite{ref:poghosyan-ueltschi, ref:jansen-kolesnikov}.

A refinement of Dobrushin's condition and that of Koteck{\'y}--Preiss for hard-core pair interactions was provided by Fern{\'a}ndez and Procacci in \cite{ref:fernandez-procacci}, see also \cite{ref:bissacot-fernandez-procacci} and, for extensions to general pairwise repulsions and stable pair potentials, \cite{ref:faris, ref:jansen-cluster, ref:jansen-kolesnikov}. 
We only use one of its corollaries in the polymer reformulation of our lattice gas below. 
Concise extensions to multi-body potentials seem out of reach so far.

\subsection{The subset polymer gas}

The polymer expansion consists in rewriting a wide variety of lattice systems, say, on $\X$, as a classical lattice gas on the countable set
\[
	\P : = \P(\X) : = \mathbf F(\X) \setminus \{\varnothing\} \subset \mathbf N(\X)
\]
of \emph{polymers} over $\X$. 
The relevant Boltzmann factor $\kappa_\P : \mathbf N(\P) \to \{0, 1\}$ is canonically given by the indicator of pairwise non-overlap, i.e.,
\[
	\kappa_\P (\mathbf P) : = \1_\mathbf F({\textstyle \sum_{j \in J} P_j}) = \1_{\{\forall j, k \in J : j \neq k \Rightarrow P_j \cap P_k = \varnothing\}} = \sum_{\mathfrak g \in \mathfrak G_J} \prod_{e \in \mathfrak g} (-\1_{\{\bigcap_{j\in e} P_j \neq \varnothing\}})
\]
for all finitely indexed polymer tuples $\mathbf P \in \P^J$, where
\[
	\mathfrak G_\V := \{\mathfrak g \subset \{e \subset \V \mid |e| = 2\}\} \subset \mathfrak H_\V
\]
denotes the set of \emph{standard graphs} on an arbitrary vertex set $\V$. 
$\kappa_\P$ clearly corresponds to the purely pairwise hard-core potential and Mayer function $V_\P : \mathbf N(\P) \to \{0,+\infty\}$, $f_\P := e^{-V_\P}-1 : \mathbf N(\P) \to \{0, -1\}$ with
\[
	V_\P(\delta_P + \delta_Q) := +\infty \1_{\{P \cap Q \neq \varnothing\}} \quad \text{and} \quad f_\P(\delta_P + \delta_Q) = -\1_{\{P \cap Q \neq \varnothing\}}
\]
for all $P, Q \in \P$. 
Since $\varnothing \notin \P$, $V_\P$ is obviously simple in our sense.

Given a collection of \emph{polymer activities} $w = (w(P))_{P \in \P}$,
we then define the a priori formal grand \emph{polymer partition function}
\[
	\Xi(w) : = Z_\P(w) : = \sum_{\mathbf n \in \mathbf N(\P)} \frac{w^\mathbf n}{\mathbf n!} \kappa_\P(\mathbf n) =\sum_{\mathfrak p \Subset \P} w^\mathfrak p \1_{\{\forall P, Q \in \mathfrak p : P \neq Q \Rightarrow P\cap Q = \varnothing\}}
\]
as well as its formal logarithm, the \emph{polymer Mayer series}
\[
	\Omega(w) : = M_\P(w) : = \sum_{\mathbf n \in \mathbf N(\P)} \frac{w^\mathbf n}{\mathbf n!} \varphi_\P(\mathbf n),
\]
where the polymer Ursell function $\varphi_\P : \mathbf N(\P) \to \Z$, defined analogously to $\varphi$, takes its classical form, i.e, for all finitely indexed tuples $\mathbf P \in \P^J$, we have
\[
	\varphi_\P(\mathbf P) = \sum_{\mathfrak h \in \mathfrak C_J} \prod_{e \in \mathfrak h} f_\P(\mathbf P_e) = \sum_{\mathfrak g \in \mathfrak{CG}_J} \prod_{e \in \mathfrak g} (-\1_{\{\bigcap_{j\in e} P_j \neq \varnothing\}}).
\]
Here and below, $\mathfrak{CG}_\V := \mathfrak G_\V \cap \mathfrak C_\V$, which coincides with the set of connected standard graphs whenever $\V\neq\varnothing$. 
Accordingly, the classical way of writing the Mayer series is
\[
	\Omega(w) = \sum_{n\in\N} \frac{1}{n!} \sum_{\mathbf P \in \P^n} \prod_{j=1}^n w(P_j) \sum_{\mathfrak g \in \mathfrak{CG}_n} \prod_{e \in \mathfrak g} (-\1_{\{\bigcap_{j\in e} P_j \neq \varnothing\}})
\]
with the obvious shorthand $\P^n = \P^{\{1,\ldots,n\}}$ and $\mathfrak{CG}_n := \mathfrak{CG}_{\{1,\ldots,n\}}$ for each $n \in \N$. 
Omitting a discussion of (reduced) polymer correlations, restriction to (polymers over) a reference volume $\Lambda \subset \X$ is implemented through
\[
	\Xi_\Lambda(w) : = \Xi(w \1_{\P(\Lambda)}) \quad \text{and} \quad \Omega_\Lambda(w) : = \Omega(w \1_{\P(\Lambda)}).
\]

\subsection{Our polymer expansion result}

For better readability, we assume that our simple potential $V : \mathbf N(\X) \to \C$ satisfies the classical assumption $V(x) = 0$ for all $x \in \X$. 
This essentially allows us to treat $V$ as a function on $\{e \Subset \X \mid |e| \geq 2\}$ and to ultimately ignore polymers of size $1$ (``monomers''). 
In this setting, we formally have
\begin{align*}
	Z_\Lambda(z) &= (1+z)^\Lambda \Xi_\Lambda(({\textstyle \frac{z}{1+z}})^\bullet \varphi \1_{\{|\bullet| \geq 2\}}) \\
	&= (1+z)^\Lambda \sum_{\mathfrak p \Subset \P} \prod_{P \in \mathfrak p} (({\textstyle \frac{z}{1+z}})^P \varphi(P) \1_{\{|P| \geq 2\}}) \1_{\{\forall P, Q \in \mathfrak p : P \neq Q \Rightarrow P\cap Q = \varnothing\}}
\end{align*}
for every $\Lambda \Subset \X$, where the sum effectively runs over all finite collections of pairwise disjoint polymers, each of size at least $2$.
For this standard result, we refer the reader to, e.g., \cite[Section~7]{ref:procacci-scoppola}. 

Our new convergence conditions and bounds for the polymer cluster expansion of our lattice gas again use the functions $\alpha : \X \to \R_+$ and $r = \frac{\alpha}{1+\alpha} : \X \to [0,1)$. 
They also feature the product from \eqref{eq:dobrushin-sokal}, here
\[
	\Gamma_A(x) := \prod_{X\Subset\X : \{x\} \subsetneq X} \max\{|\e^{-V(X)}|, 1 + |f(X)| A^S \mid \varnothing \neq S \subset X \setminus \{x\}\},
\]
as a function $\Gamma_A : \X \to [1,\infty]$ depending on an auxiliary $A : \X \to \R_+$, allowed to differ from $\alpha$. 
The hypothesis of Theorem~\ref{thm:dobrushin-sokal} then reads
\[
	|z| \Gamma_\alpha \leq r \quad \text{or} \quad |z| (1+\alpha) \Gamma_\alpha \leq \alpha,
\]
noting that one can absorb $\e^{-V(x)}$ into $z(x)$ for all $x \in \X$ there.

\begin{theorem} \label{thm:dobrushin-sokal-polymer}
	Let $z : \X \to \C \setminus \{-1\}$ and suppose that $V$ is a simple interaction potential satisfying $V(x) = 0$ for all $x \in \X$ as well as
	\be \label{eq:dobrushin-sokal-polymer}
		|{\textstyle \frac{z}{1+z}}| ((1+A)\Gamma_A - A) \leq A \quad \text{and} \quad A {\textstyle \frac{\Gamma_A - 1}{\Gamma_A}} \leq \alpha
	\ee
	for some $A : \X \to \R_+$. 
	Let further $w : \P \to \C$, $P \mapsto (\frac{z}{1+z})^P \varphi(P) \1_{\{|P|\geq 2\}}$. 
	Then, for all simultaneously $|z|$-finite and $r$-finite $\Lambda \subset \X$, we have
	\[
		\frac{Z_\Lambda(z)}{(1+z)^\Lambda} = \exp\left( \sum_{n \in \N} \frac{1}{n!} \sum_{\mathbf P \in \P(\Lambda)^n} \prod_{j=1}^n w(P_j) \sum_{\mathfrak g \in \mathfrak{CG}_n} \prod_{e\in\mathfrak g} (-\1_{\{\bigcap_{j\in e} P_j = \varnothing\}}) \right) \neq 0,
	\]
	i.e., $Z_\Lambda(z) / (1+z)^\Lambda = \Xi_\Lambda(w) = \exp(\Omega_\Lambda(w)) \in \C \setminus \{0\}$. 
	More precisely, for every $r$-finite $\Lambda \subset \X$, one has
	\[
		0 < (1-r)^\Lambda \leq \Xi_\Lambda(-|w|) \leq |\Xi_\Lambda(w)| \leq \Xi_\Lambda(|w|) \leq (1+r)^\Lambda < \infty,
	\]
	\[
		|\Omega_\Lambda(w)| \leq \sum_{\mathbf n \in \mathbf N(\P(\Lambda))} \frac{|w|^\mathbf n}{\mathbf n!} |\varphi_\P(\mathbf n)| = -\Omega_\Lambda(-|w|) \leq - \log (1 - r)[\Lambda] < \infty.
	\]
\end{theorem}

Observe first that the latter theorem makes no claims about the absolute convergence of the a priori cluster expansions $M_\Lambda(z)$ and $\rho_\Lambda(\bullet, z)$ for any $\Lambda \subset \X$. 
However, the bounds at the end of Theorem~\ref{thm:dobrushin-sokal-polymer} are still analogous to those inferred in Corollary~\ref{cor:Z-M-r-bounds}.  
Using condition \eqref{eq:dobrushin-sokal-polymer}, where taking $A=\alpha$ makes the second inequality trivial, over Theorem~\ref{thm:dobrushin-sokal}'s $|z| (1+\alpha) \Gamma_\alpha \leq \alpha$ is most clearly advantageous in the case of a ``physical'' activity $z : \X \to \R_+$.

\subsection{Relation to existing results}

For comparisons against more classical results, first consider a physical potential $V$ that implicitly scales with an \emph{inverse temperature} $\beta \in \R_+$ in the sense that $V(X) = V_\beta(X) = \beta V_1(X) \in \R$ for all $X \Subset \X$. 
If $A : \X \to \R_+$ is such that $\Gamma_A < \infty$, then
\[
	1 \leq (1+A) \Gamma_A - A = \Gamma_A + A (\Gamma_A - 1) \to 1 \quad \text{as} \quad \beta \downarrow 0
\]
so, in the high-temperature regime, the more relevant first inequality in \eqref{eq:dobrushin-sokal-polymer} approaches a requirement like $|\frac{z}{1+z}| < A$, provided $A$ and the above convergence are somewhat uniform over $\X$. 
If additionally $A > 1$, then, for sufficiently small $\beta$, the activity region defined by \eqref{eq:dobrushin-sokal-polymer} may therefore encompass all physical choices $z : \X \to \R_+$, indicating an absence of phase transitions. 
With this in mind, we define the function
\[
	D_A : \X \to \overline{\R}_+, \quad x \mapsto \sum_{X \Subset \X : \{x\} \subsetneq X} |V(X)| A^X,
\]
for $A : \X \to [1,\infty)$, which directly inherits the $\beta$-scaling from $V$.

\begin{cor} \label{cor:gallavotti-miracle-sole-robinson-refined}
	Let $z : \X \to \C \setminus \{-1\}$ and suppose that $V$ is a simple interaction potential satisfying $V(x) = 0$ for all $x \in \X$ as well as
	\[
		|{\textstyle \frac{z}{1+z}}| ((1 + A) \e^{D_A/A} - A) \leq A \quad \text{and} \quad A (1-\e^{-D_A / A}) \leq \alpha
	\]
	for some $A : \X \to [1,\infty)$. 
	Let further $w = (\frac{z}{1+z})^\bullet \varphi \1_{\{|\bullet| \geq 2\}} : \P \to \C$. 
	Then the conclusions of Theorem~\ref{thm:dobrushin-sokal-polymer} hold.
\end{cor}

\begin{proof}
	Observe that, for all $A : \X \to [1,\infty)$ and every $x \in \X$, we have 
	\[
		\Gamma_A(x) = \prod_{X \Subset \X : \{x\} \subsetneq X} (1+|\e^{-V(X)}-1| A^{X \setminus \{x\}}) \leq \e^{\sum_{X \Subset \X : \{x\} \subsetneq X} |V(X)| A^{X \setminus \{x\}}},
	\]
	the sum in the last exponent being $D_A(x)/A(x)$, and therefore also
	\[
		{\textstyle \frac{\Gamma_A-1}{\Gamma_A}} = 1 - 1/\Gamma_A \leq 1 - \e^{-D_A / A}.
	\]
	Hence, the corollary's hypothesis implies that of Theorem~\ref{thm:dobrushin-sokal-polymer}.
\end{proof}

Corollary~\ref{cor:gallavotti-miracle-sole-robinson-refined} directly improves upon \cite[Theorem~5]{ref:procacci-scoppola} and, although the exact applications are different, also refines the condition in \cite[Theorem~4.2.7]{ref:ruelle}, originally from \cite{ref:gallavotti-miracle-sole-robinson} and itself a slight improvement to Gallavotti and Miracle-Sol{\'e}'s \cite{ref:gallavotti-miracle-sole}. 
These results would respectively require
\[
	|{\textstyle \frac{z}{1+z}}| \exp(4 {D_1} \e^{D_1}) < 1 \quad \text{and} \quad |{\textstyle \frac{z \e^c}{1+z\e^c}}| (2 \exp(\e^{D_1} - 1) - 1) < 1
\]
with $c : \X \to \overline{\R}_+$, $x \mapsto \sum_{X \Subset \X : \{x\} \subsetneq X} \max\{-\Re V(X), 0\}$, as before.

Lastly, let us compare our Theorem~\ref{thm:dobrushin-sokal-polymer} against the specialisation of Nguyen and Fern{\'a}ndez's recent article \cite{ref:nguyen-fernandez} to our lattice gas setup. 
Via the polymer reformulation and the Gruber--Kunz condition, see Proposition~\ref{prop:gruber-kunz}, they essentially derived a criterion requiring $|\frac{z}{1+z}| \leq 1$ as well as
\be \label{eq:nguyen-fernandez}
	|f(e)| \prod_{x \in e} \left( 1 + \sum_{e' \Subset \X : x \in e'} \tau(e') \right) \prod_{e' \Subset \X : e \cap e' \neq \varnothing} (1 + \tau(e')) \leq \tau(e)
\ee
for all non-empty $e \Subset \X$, cf.\ \cite[(107)]{ref:nguyen-fernandez}. 
This is then fashioned into more concrete criteria in terms of the norms $\|V\|_A := \sup \{D_A(x) \mid x\in\X\}$ with $A : \X \to [1,\infty)$. 
For comparison's sake, these criteria are less relevant than an intermediate estimate in their derivation, namely
\[
	\prod_{e' \Subset \X : e \cap e' \neq \varnothing} (1 + \tau(e')) \leq \prod_{x \in e} \prod_{e' \Subset \X : x \in e'} (1 + \tau(e'))
\]
with $x \in e \Subset \X$, cf.\ the middle term in \cite[(109)]{ref:nguyen-fernandez}. 
Replacing the edge-indexed product in this way and setting $\tau(e) := |f(e)| A^e$ for all $e \Subset \X$ with $A \geq 1$ turns \eqref{eq:nguyen-fernandez} into the requirement
\[
	\left( 1 + \sum_{e \Subset \X : x \in e} |f(e)| A^e \right) \prod_{e \Subset \X : x \in e} (1 + |f(e)| A^e) \leq A(x)
\]
for all $x \in \X$. 
Meanwhile, our condition \eqref{eq:dobrushin-sokal-polymer} can be met whenever
\[
	(1 + A(x)) \prod_{e \Subset \X : \{x\} \subsetneq e} (1 + |f(e)| A^{e \setminus \{x\}}) \leq 2 A(x)
\]
with $V(x) = 0$ for all $x \in \X$ and $|\frac{z}{1+z}| \leq 1 \leq A$. 

Roughly speaking, Nguyen and Fern{\'a}ndez's approach in \cite{ref:nguyen-fernandez} may be better for sufficiently weak interactions (high temperature) and is applicable in far broader generality. 
Our method demonstrates how to allow for stronger and, by dropping the requirement $A\geq 1$, even hard-core interactions.

\section{Preliminaries on power series and coefficient operations} \label{sec:prelims-convolution-poisson}

Our power series are mostly of the form
\[
	\pi^z[g] : = \sum_{\mathbf n \in \mathbf N} \frac{z^\mathbf n}{\mathbf n!} g(\mathbf n) = \sum_{n\in\N_0} \frac{1}{n!} \sum_{\mathbf x \in \X^n} \prod_{j=1}^n z(x_j) g(\mathbf x)
\]
where the coefficient $g : \mathbf N \to \A$ takes values in $\A = \C$ or $\A = \overline \R_+$ or sets of formal power series with coefficients in $\C$ or $\overline\R_+$. 
The collection $z = (z(x))_{x \in \X}$ of a priori formal variables may be evaluated as a function $\X \to \A$ in the absence of convergence issues. 
We think of $\pi^z[\bullet]$ as the integral with respect to a non-normalised multivariate \emph{Poisson distribution} with intensity $z$ and employ measure-theoretic conventions when evaluating $z$ more concretely.  
We do not claim any novelty for the contents of this section which are easily extrapolated from sections in, e.g., \cite{ref:ruelle, ref:gallavotti-miracle-sole}. 
Our self-contained reiteration is streamlined to fit our notation.

\subsection{Poisson convolution}

Given $g, h : \mathbf N \to \A$, we can define their \emph{convolution} $g \ast h : \mathbf N \to \A$ by setting
\[
	g \ast h(\mathbf n) := \sum_{\mathbf k \in \mathbf N} \binom{\mathbf n}{\mathbf k} g(\mathbf k) h(\mathbf n - \mathbf k) \quad \text{or} \quad (g \ast h)(\mathbf x) := \sum_{K \subset J} g(\mathbf x_K) h(\mathbf x_{J \setminus K})
\]
for all $\mathbf n \in \mathbf N$ or $\mathbf x \in \X^J$ with finite index set $J$. 
Recalling that
\[
	\binom{\mathbf n}{\mathbf k} : = \prod_{x \in \X} \binom{n_x}{k_x} = \prod_{x \in \X} \frac{n_x!}{k_x!} \1_{\{k_x \leq n_x\}} = \frac{\mathbf n!}{\mathbf k!} \1_{\{\mathbf k \leq \mathbf n\}}
\]
is the multivariate binomial coefficient for $\mathbf n = (n_x)_{x \in \X}, \mathbf k = (k_x)_{x \in \X} \in \mathbf N$, the two definitions are equivalent to the formal identity
\be \label{eq:pi-z-g-ast-h}
	\pi^z[g \ast h] = \pi^z[g] \cdot \pi^z[h].
\ee
In particular, the $\ast$-operation is also associative, commutative and $\A$-bilinear.

More broadly, if $L$ is an arbitrary finite set and $g_l : \mathbf N \to \A$ for every $l \in L$, the convolution $\Conv_{l \in L} g_l : \mathbf N \to \A$ is given by
\[
	\mathbf n \mapsto \sum_{(\mathbf k^{(l)})_{l \in L} \in \mathbf N^L} \binom{\mathbf n}{(\mathbf k^{(l)})_{l \in L}} \prod_{l \in L} g_l(\mathbf k^{(l)}) \quad \text{or} \quad \mathbf x \mapsto \sum_{\mathbf P \in \mathfrak P_J^{(L)}} \prod_{l \in L} g_l(\mathbf x_{P_l})
\]
with $\mathbf n\in\mathbf N$ or finitely indexed $\mathbf x \in \X^J$, respectively, where
\begin{align*}
	\binom{\mathbf n}{(\mathbf k^{(l)})_{l \in L}} : = \prod_{x \in \X} \binom{n_x}{(k^{(l)}_x)_{l \in L}} = \frac{\mathbf n!}{\prod_{l \in L} \mathbf k^{(l)}!} \1_{\{\sum_{l \in L} \mathbf k^{(l)} = \mathbf n\}}
\end{align*}
denotes the multivariate multinomial coefficient and
\[
	\mathfrak P_J^{(L)} : = \{{\textstyle \mathbf P \in \{S \subset J\}^L \mid J = \bigsqcup_{L \in L} P_l}\}
\]
the set of \emph{$L$-indexed partitions} of $J$ into pairwise disjoint but \emph{possibly empty} subsets. 
In the above situation, we likewise formally obtain
\be \label{eq:pi-z-conv-g}
	\pi^z[{\textstyle \Conv_{l \in L} g_l}] = \prod_{l \in L} \pi^z[g_l].
\ee
If $g_l = g : \mathbf N \to \A$ for every $l \in L$, we also write $g^{\ast |L|} = \Conv_{l \in L} g_l$. 
To include the case $L = \varnothing$, we define the empty convolution as $1_\ast : = \1_{\{0\}} : \mathbf N \to \A$, the $\ast$-identity element.

Note also, that for every finite set $L$ and all $g : \mathbf N^L \to \A$, we can define
\be \label{eq:conv-g-generalised}
	G : \mathbf N \to \A, \quad \mathbf n \mapsto \sum_{(\mathbf k^{(l)})_{l \in L} \in \mathbf N^L} \binom{\mathbf n}{(\mathbf k^{(l)})_{l \in L}} g((\mathbf k^{(l)})_{l \in L}),
\ee
i.e., $G = \sum_{(\mathbf k^{(l)})_{l \in L} \in \mathbf N^L} g((\mathbf k^{(l)})_{l \in L}) \Conv_{l \in L} \1_{\{\mathbf k^{(l)}\}}$, and formally, by \eqref{eq:pi-z-conv-g},
\be \label{eq:pi-z-conv-g-generalised}
	\pi^z[G] = \sum_{(\mathbf k^{(l)})_{l \in L} \in \mathbf N^L} \left( \prod_{l \in L} \frac{z^{\mathbf k^{(l)}}}{\mathbf k^{(l)}!} \right) g((\mathbf k^{(l)})_{l \in L}).
\ee
With the implicit measure-theoretic conventions, monotone and dominated convergence yield the following.

\begin{lemma} \label{lem:pi-z-conv-g-generalised}
	Let $z : \X \to \A$, $L$ a finite set, $g : \mathbf N^L \to \A$ and $G : \mathbf N \to \A$ given by \eqref{eq:conv-g-generalised}. 
	Then we have the following:
	\begin{itemize}
		\item[(i)] If $\A = \overline \R_+$, then \eqref{eq:pi-z-conv-g-generalised} holds unconditionally as an identity in $\overline \R_+$.
		
		\item[(ii)] If $\A = \C$, then \eqref{eq:pi-z-conv-g-generalised} holds as an identity in $\C$ whenever
		\[
			\sum_{(\mathbf k^{(l)})_{l \in L} \in \mathbf N^L} \left( \prod_{l \in L} \frac{|z|^{\mathbf k^{(l)}}}{\mathbf k^{(l)}!} \right) |g((\mathbf k^{(l)})_{l \in L})| < \infty.
		\]
	\end{itemize}
\end{lemma}

\subsection{Convolution power series as coefficients}

Given $g : \mathbf N \to \A$ with $g(0) = 0$, we can define the \emph{$\ast$-exponential}
\[
	\e^{\ast g} : = \exp_\ast g : = \sum_{m \in \N_0} \frac{g^{\ast m}}{m!} : \mathbf N \to \A
\]
since, for every finitely indexed tuple $\mathbf x \in \X^J$, we have
\[
	\sum_{m \in \N_0} \frac{g^{\ast m}}{m!}(\mathbf x) = \sum_{m \in \N_0} \frac{1}{m!} \sum_{\mathbf P \in \mathfrak P_J^{(\{1, \ldots, m\})}} \prod_{l=1}^m g(\mathbf x_{P_l}) = \sum_{\mathfrak p \in \mathfrak P_J} \prod_{P \in \mathfrak p} g(\mathbf x_P),
\]
the last sum and product having only finite index sets. 
More explicitly, $g(0) = 0$ makes the inner sum in the middle effectively run over ordered partitions of $J$ into $m$ non-empty subsets and discarding the ordering cancels the factor $1/m!$. 
Convergence being a non-issue, the usual argument yields
\[
	\e^{\ast \sum_{l \in L} g_l} = \Conv_{l \in L} \e^{\ast g_l}
\]
whenever $L$ is a finite index set and $g_l : \mathbf N \to \A$ with $g_l(0) = 0$ for all $l \in L$.

If, e.g., $\A = \C$, then, by the same token, we can define the \emph{$\ast$-logarithm}
\[
	\log_\ast h = \log_\ast(1_\ast + g) : = \sum_{m \in \N} \frac{(-1)^{m-1} g^{\ast m}}{m} : \mathbf N \to \A
\]
and the \emph{$\ast$-reciprocal}, defined via the geometric series,
\[
	h^{\ast -1} = (1_\ast + g)^{\ast -1} : = \sum_{m \in \N_0} (-1)^m g^{\ast m} : \mathbf N \to \A
\]
of $h = 1_\ast + g : \mathbf N \to \A$ with $h(0) = 1$. 
Again, all relevant sums implicitly run over finite index sets when evaluating on some $\mathbf n \in \mathbf N$ so the usual arguments also yield the expected behaviour:
\[
	\log_\ast \e^{\ast g} = g, \quad \e^{\ast \log_\ast h} = h \quad \text{and} \quad (h^{\ast -1}) \ast h = h \ast (h^{\ast -1}) = 1_\ast
\]
as well as
\[
	\log_\ast({\textstyle \Conv_{l \in L} h_l}) = \sum_{l \in L} \log_\ast h_l \quad \text{and} \quad \left( \Conv_{l \in L} h_l \right)^{\ast -1} = \Conv_{l \in L} h_l^{\ast -1}
\]
whenever $L$ is a finite index set and $h_l : \mathbf N \to \A$ with $h_l(0) = 1$ for all $l \in L$. 
Note, that for every finitely indexed $\mathbf x \in \X^J$ with $J \neq \varnothing$, we explicitly obtain
\[
	(\log_\ast h)(\mathbf x) = \sum_{\mathfrak p \in \mathfrak P_J} (-1)^{|\mathfrak p| - 1} (|\mathfrak p| - 1)! \prod_{P \in \mathfrak p} h(\mathbf x_P),
\]
a known M{\"o}bius inversion formula along partitions, following the same argument as above.

Since the above series, with $g : \mathbf N \to \A$, $g(0) = 0$ and $h = 1_\ast + g$, are pointwise finite sums, \eqref{eq:pi-z-conv-g} formally yields
\[
	\pi^z[\e^{\ast g}] = \pi^z[\exp_\ast g] = \sum_{m \in \N_0} \frac{\pi^z[g]^m}{m!} = \exp(\pi^z[g]) = \e^{\pi^z[g]}
\]
and, provided $\A$ is suitable, also
\[
	\pi^z[\log_\ast h] = \sum_{m \in \N} \frac{(-1)^{m-1} \pi^z[g]^m}{m} = \log(1 + \pi^z[g]) = \log(\pi^z[h])
\]
and $\pi^z[h^{\ast -1}] = \sum_{m \in \N_0} (-1)^m \pi^z[g]^m = (1 + \pi^z[g])^{-1} = \pi^z[h]^{-1}$. 
The validity of the latter identities upon evaluation over $\A = \C$ is mediated by the case $\A = \overline \R_+$ via monotone and dominated convergence and Lemma~\ref{lem:pi-z-conv-g-generalised}. 
Treating the first of the above convolution power series
is quite straightforward, given the exponential series's positive coefficients and infinite radius of convergence over $\C$ as well as the convention $\e^{+\infty} = +\infty$. 
For the logarithmic and reciprocal/geometric cases, one needs to observe that
\[
	\log(1 \pm \zeta) = \sum_{m \in \N} \frac{(\mp 1)^{m-1} \zeta^m}{m} \quad \text{and} \quad (1 \pm \zeta)^{-1} = \sum_{m \in \N_0} (\mp 1)^m \zeta^m
\]
as well-behaved identities in $\C$ necessitate $|\zeta| < 1$.

\subsection{Derivation via shifting}

By viewing $\pi^z[g] = \sum_{\mathbf n \in \mathbf N} \frac{z^\mathbf n}{\mathbf n!} g(\mathbf n)$, for given $g : \mathbf N \to \A$, as an abstract Taylor series, one effortlessly sees that
\[
	\frac{\partial}{\partial z(x)} \pi^z[g] = \sum_{\mathbf n \in \mathbf N} \frac{z^\mathbf n}{\mathbf n!} g(\delta_x + \mathbf n) = \pi^z[D^x g]
\]
formally holds with $D^x g : = g(\delta_x + \bullet)$ for every $x \in \X$. 
The Leibniz product rule for the partial derivatives on the left-hand side directly translates into
\be \label{eq:D-g-ast-h}
	D^x (g \ast h) = (D^x g) \ast h + g \ast (D^x h)
\ee
for arbitrary $x \in \X$ and $g, h : \mathbf N \to \A$. 
More generally, for every $\mathbf m \in \mathbf N$ and $g : \mathbf N \to \A$, we can take formal derivatives of higher-order to obtain
\[
	\frac{\partial^{|\mathbf m|}}{(\partial z)^\mathbf m} \pi^z[g] = \sum_{\mathbf n \in \mathbf N} \frac{z^\mathbf n}{\mathbf n!} g(\mathbf m + \mathbf n) = \pi^z[D^\mathbf m g]
\]
with $D^\mathbf m g : = g(\mathbf m + \bullet)$ and, given a finite set $L$ and $g_l : \mathbf N \to \A$ for every $l \in L$, \eqref{eq:D-g-ast-h} generalises to
\be \label{eq:D-conv-g}
	D^\mathbf m \Conv_{l \in L} g_l = \sum_{(\mathbf k^{(l)})_{l \in L} \in \mathbf N^L} \binom{\mathbf m}{(\mathbf k^{(l)})_{l \in L}} \Conv_{l \in L} D^{\mathbf k^{(l)}} g_l.
\ee
Note that the latter can be read as a convolution formula with respect to $\mathbf m$ if $\A$ and its innate multiplication are replaced by $\{g : \mathbf N \to \A\}$ and $\ast$, respectively, and that substituting a finitely indexed tuple $\mathbf x \in \X^J$ for $\mathbf m$ would allow writing the right-hand side of \eqref{eq:D-conv-g} as a sum over $\mathfrak P_J^{(L)}$.

Of particular relevance is the interaction of this derivation with $\ast$-exponentials and $\ast$-logarithms. 
Namely, let $g : \mathbf N \to \A$ with $g(0) = 0$ and verify that $D^x \e^{\ast g} = \e^{\ast g} \ast (D^x g)$ for every $x \in \X$ and, more generally,
\be \label{eq:D-exp}
	D^\mathbf x \e^{\ast g} = \e^{\ast g} \ast \sum_{\mathfrak p \in \mathfrak P_J} \Conv_{P \in \mathfrak p} D^{\mathbf x_P} g
\ee
for every finitely indexed tuple $\mathbf x \in \X^J$. 
If, e.g., $\A = \C$ and $h = 1_\ast + g$, analogous instances of the chain rule and Fa\`a di Bruno's formula include $D^x \log_\ast h = h^{\ast -1} \ast (D^x h)$ for every $x \in \X$ as well as
\begin{equation} \label{eq:D-log}
	D^\mathbf x \log_\ast h = \sum_{\mathfrak p \in \mathfrak P_J} (-1)^{|\mathfrak p| - 1} (|\mathfrak p| - 1)! \Conv_{P \in \mathfrak p} (h^{\ast -1} \ast D^{\mathbf x_P} h)
\end{equation}
whenever $\mathbf x \in \X^J$ with non-empty finite index set $J$.

The derivation operators also provide us with a nice way of framing the superposition principle, which essentially says that Poisson distributions form a measure-theoretic convolution semigroup with respect to the addition of their intensities. 
Given $g : \mathbf N \to \A$, we collect the values
\[
	\Pi^z[g](\mathbf m) : = \pi^z[D^\mathbf m g], \quad \mathbf m \in \mathbf N,
\]
into a formally well-defined function $\Pi^z[g]$ that, by \eqref{eq:D-conv-g} and \eqref{eq:pi-z-conv-g}, satisfies
\[
	\Pi^z[{\textstyle \Conv_{l \in L} g_l}] = \Conv_{l \in L} \Pi^z[g_l]
\]
whenever $L$ is a finite set and $g_l : \mathbf N \to \A$ for all $l \in L$. 
Analogously, \eqref{eq:D-exp} formally yields
\[
	\Pi^z[\e^{\ast g}] = \e^{\pi^z[g]} \cdot \exp_\ast(\Pi^z[g] \1_{\mathbf N \setminus \{0\}})
\]
whenever $g : \mathbf N \to \A$ with $g(0) = 0$, and, if, additionally, $\A$ is suitable and $h = 1_\ast + g$, then \eqref{eq:D-log} translates into
\[
	\Pi^z[\log_\ast h] = \log(\pi^z[h]) \cdot 1_\ast + \log_\ast (\pi^z[h]^{-1} \cdot \Pi^z[h]).
\]
The aforementioned superposition principle can now be phrased as follows: 
if $\tilde z = (\tilde z(x))_{x \in \X}$ is another collection of abstract variables commuting among each other and with those in $z$, then
\be \label{eq:Pi-z-semigroup}
	\Pi^{z + \tilde z}[g](\mathbf m) = \Pi^z[\Pi^{\tilde z}[g]](\mathbf m)
\ee
formally holds for all $g : \mathbf N \to \A$, $\mathbf m \in \mathbf N$. 
For proving the latter identity, it suffices to consider the case $\mathbf m = 0$ since clearly $\Pi^z[g](\mathbf m) = \Pi^z[D^\mathbf m g](0)$ by definition. 
For $\mathbf m = 0$, \eqref{eq:Pi-z-semigroup} reads
\[
	\pi^{z + \tilde z}[g] = \sum_{\mathbf n \in \mathbf N} \frac{(z + \tilde z)^\mathbf n}{\mathbf n!} g(\mathbf n) = \sum_{\mathbf k \in \mathbf N} \frac{z^\mathbf k}{\mathbf k!} \sum_{\tilde{\mathbf k} \in \mathbf N} \frac{\tilde z^{\tilde{\mathbf k}}}{\tilde{\mathbf k}!} g(\mathbf k + \tilde{\mathbf k}) = \pi^z[\Pi^{\tilde z}[g]]
\]
and easily follows from the multivariate binomial formula
\[
	(z + \tilde{z})^\bullet = z^\bullet \ast \tilde{z}^\bullet : \mathbf n \mapsto (z + \tilde z)^\mathbf n = \sum_{\mathbf k \in \mathbf N} \binom{\mathbf n}{\mathbf k} z^\mathbf k \tilde{z}^{\mathbf n - \mathbf k}.
\]
As usual, monotone and dominated convergence allow for according identities after evaluating in the cases $\A = \overline \R_+$ and $\A = \C$.

\begin{lemma} \label{lem:superposition-principle}
	Let $z, \tilde{z} : \X \to \A$, $g : \mathbf N \to \A$ and $\mathbf m \in \mathbf N$. 
	Then we have the following:
	\begin{itemize}
		\item[(i)] If $\A = \overline \R_+$, then \eqref{eq:Pi-z-semigroup} holds unconditionally as an identity in $\overline \R_+$.
		
		\item[(ii)] If $\A = \C$, then \eqref{eq:Pi-z-semigroup} holds as an identity in $\C$ whenever
		\[
			\Pi^{|z| + |\tilde{z}|}[|g|](\mathbf m) = \Pi^{|z|}[\Pi^{|\tilde{z}|}[|g|]](\mathbf m) < \infty.
		\]
	\end{itemize}
\end{lemma}

As indicated above, the binomial formula is a special instance of our Poisson convolution so the superposition principle, here \eqref{eq:Pi-z-semigroup} / Lemma~\ref{lem:superposition-principle}, and \eqref{eq:pi-z-conv-g-generalised} / Lemma~\ref{lem:pi-z-conv-g-generalised} turn out to be two sides of the same coin.

\section{General correlations and the Kirkwood--Salsburg hierarchy} \label{sec:prelims-conditional-root-KS}

We extend the potential $V$, the Boltzmann factor $\kappa$ and the Ursell function $\varphi$ to functions on $\mathbf N \times \mathbf N$, where the second argument, readily omitted if equal to the mutually identified vacuum configurations $0 \in \mathbf N$, $() \in \X^\varnothing = \X^0$ or $\varnothing \in \mathbf F$, is usually called the \emph{boundary condition}. 
Namely, let
\[
	V : \mathbf N \times \mathbf N \to \C \cup \{+\infty\}, \quad (\mathbf n, \mathbf b) \mapsto V(\mathbf n \mid \mathbf b) : =  \sum_{\mathbf c \in \mathbf N} \binom{\mathbf b}{\mathbf c} V(\mathbf n + \mathbf c),
\]
i.e., $V(\mathbf x \mid \mathbf y) = \sum_{S \subset K} V(\mathbf x + \mathbf y_S)$ for all finitely indexed tuples $\mathbf x \in \X^J$, $\mathbf y \in \X^K$, where the sum of tuples corresponds to their concatenation. The conditional Boltzmann factor is then
\[
	\kappa : \mathbf N \times \mathbf N \to \C, \quad (\mathbf n, \mathbf b) \mapsto \kappa(\mathbf n \mid \mathbf b) : = \e^{- \sum_{\mathbf k \in \mathbf N \setminus \{0\}} \binom{\mathbf n}{\mathbf k} V(\mathbf k \mid \mathbf b)}
\]
An important observation here is the \emph{conditional multiplicativity}
\be \label{eq:D-kappa}
	\kappa(\mathbf m + \mathbf n \mid \mathbf b) = \kappa(\mathbf m \mid \mathbf n + \mathbf b) \kappa(\mathbf n \mid \mathbf b)
\ee
for all $\mathbf m, \mathbf n, \mathbf b \in \mathbf N$. 
Note also that, for all $\mathbf b \in \mathbf N$, we have $\kappa(0 \mid \mathbf b) = 1$ and can therefore extend and rewrite \eqref{eq:kappa-sum-prod-phi} in the forms
\be \label{eq:kappa-exp-phi-phi-log-kappa}
	\kappa(\bullet \mid \mathbf b) = \exp_\ast \varphi(\bullet \mid \mathbf b) \quad \text{and} \quad \varphi(\bullet \mid \mathbf b) : = \log_\ast \kappa(\bullet \mid \mathbf b),
\ee
where $\ast$-exponential and $\ast$-logarithm only act on the respective unspecified first arguments here and the second identity is used as the definition of the conditional Ursell function $\varphi : \mathbf N \times \mathbf N \to \C$.

Recall that we assume simplicity of $V$, meaning $\kappa(\bullet \mid \mathbf b)$ vanishes outside
\[
	\mathbf F_{\mid \mathbf b} : = \{X \Subset \X \mid \mathbf b[X] = 0\} \subset \mathbf F_{\mid 0} = \mathbf F \subset \mathbf N
\]
for any fixed $\mathbf b \in \mathbf N$, i.e., that $\kappa(x \mid x) = \e^{-V(x) - V(2 \delta_x)} = 0$ for all $x \in \X$.

\subsection{Convolutional Boltzmann factors}

For all $\mathbf b \in \mathbf N$, the $\ast$-reciprocal of $\kappa(\bullet \mid \mathbf b) = 1_\ast + \kappa(\bullet \mid \mathbf b) \1_{\mathbf N \setminus \{0\}}$ exists and we can also define
\[
	\psi(\mathbf m, \bullet \mid \mathbf b) : = \kappa(\bullet \mid \mathbf b)^{\ast -1} \ast D^\mathbf m \kappa(\bullet \mid \mathbf b)
\]
for every $\mathbf m \in \mathbf N$, yielding a trivariate function $\psi : \mathbf N \times \mathbf N \times \mathbf N \to \C$. 
Not only is $\psi(0, \bullet \mid \mathbf b) = 1_\ast$ but \eqref{eq:D-kappa} extends to the convolutional version
\be \label{eq:psi-conditional-convolution}
	\psi(\mathbf m + \mathbf n, \bullet \mid \mathbf b) = \psi(\mathbf m, \bullet \mid \mathbf n + \mathbf b) \ast \psi(\mathbf n, \bullet \mid \mathbf b)
\ee
for any given $\mathbf m, \mathbf n, \mathbf b \in \mathbf N$. 
To see this, use \eqref{eq:D-kappa} twice to write
\[
	\psi(\mathbf m + \mathbf n, \bullet \mid \mathbf b) = \kappa(\mathbf m \mid \mathbf n + \mathbf b) \kappa(\mathbf n \mid \mathbf b) \cdot \kappa(\bullet \mid \mathbf m + \mathbf n + \mathbf b) \ast \kappa(\bullet \mid \mathbf b)^{\ast -1},
\]
convolve with $1_\ast = \kappa(\bullet \mid \mathbf n + \mathbf b)^{\ast -1} \ast \kappa(\bullet \mid \mathbf n + \mathbf b)$ and again use \eqref{eq:D-kappa} twice as
\[
	\kappa(\mathbf m \mid \mathbf n + \mathbf b) \kappa(\bullet \mid \mathbf m + \mathbf n + \mathbf b) = D^\mathbf m \kappa(\bullet \mid \mathbf n + \mathbf b)
\]
and $\kappa(\mathbf n \mid \mathbf b) \kappa(\bullet \mid \mathbf n + \mathbf b) = D^\mathbf n \kappa(\bullet \mid \mathbf b)$.

Next, we relate $\psi$ and $\varphi$. 
Using \eqref{eq:kappa-exp-phi-phi-log-kappa} in \eqref{eq:D-exp}, we obtain
\begin{align*}
	\psi(\mathbf x, \bullet \mid \mathbf y) & = \kappa(\bullet \mid \mathbf y)^{\ast -1} \ast \kappa(\bullet \mid \mathbf y) \ast \sum_{\mathfrak p \in \mathfrak P_J} \Conv_{P \in \mathfrak p} D^{\mathbf x_P} \varphi(\bullet \mid \mathbf y) \\
	& = \sum_{\mathfrak p \in \mathfrak P_J} \Conv_{P \in \mathfrak p} D^{\mathbf x_P} \varphi(\bullet \mid \mathbf y)
\end{align*}
for all finitely indexed tuples $\mathbf x \in \X^J$, $\mathbf y \in \X^K$ and, if additionally $J \neq \varnothing$, plugging \eqref{eq:kappa-exp-phi-phi-log-kappa} into \eqref{eq:D-log} yields the M{\"o}bius inverse formula
\begin{align}
	D^\mathbf x \varphi(\bullet \mid \mathbf y) & = \sum_{\mathfrak p \in \mathfrak P_J} (-1)^{|\mathfrak p| - 1} (|\mathfrak p| - 1)! \Conv_{P \in \mathfrak p} (\kappa(\bullet \mid \mathbf y)^{\ast -1} \ast D^{\mathbf x_P} \kappa(\bullet \mid \mathbf y)) \notag \\
	& = \sum_{\mathfrak p \in \mathfrak P_J} (-1)^{|\mathfrak p| - 1} (|\mathfrak p| - 1)! \Conv_{P \in \mathfrak p} \psi(\mathbf x_P, \bullet \mid \mathbf y). \label{eq:D-phi-sum-conv-psi}
\end{align}
In particular, for all $x \in \X$ and $\mathbf b \in \mathbf N$, we have
\be \label{eq:psi-x-D-x-phi}
	\psi(\delta_x, \bullet \mid \mathbf b) = D^x \varphi(\bullet \mid \mathbf b)
\ee
and, by simplicity of $V$, \eqref{eq:D-phi-sum-conv-psi} yields
\be \label{eq:D-phi-hard-core-clique}
	D^{m\delta_x} \varphi(\bullet \mid \mathbf b) = (-1)^{m-1} (m-1)! \psi(\delta_x, \bullet \mid \mathbf b)^{\ast m}.
\ee

\subsection{Controlling partition functions and Mayer series via effective activities}

Given $\mathbf m, \mathbf b \in \mathbf N$, we now formally define
\begin{align*}
	Z(\mathbf m, z \mid \mathbf b) & : = \pi^z[D^\mathbf m \kappa(\bullet \mid \mathbf b)] = \Pi^z[\kappa(\bullet \mid \mathbf b)](\mathbf m), \\
	M(\mathbf m, z \mid \mathbf b) & : = \pi^z[D^\mathbf m \varphi(\bullet \mid \mathbf b)] = \Pi^z[\varphi(\bullet \mid \mathbf b)](\mathbf m), \\
	\rho(\mathbf m, z \mid \mathbf b) & : = \pi^z[\psi(\mathbf m, \bullet \mid \mathbf b)]
\end{align*}
and call these the $\mathbf m$-rooted $\mathbf b$-conditional partition function, Mayer series and correlation, respectively. 
Independently of one another, we may notationally suppress not only the boundary condition $\mathbf b = 0$ but also the root configuration $\mathbf m = 0$. 
In particular, Section~\ref{sec:results} is consistent with
\[
	Z(z) = Z(0, z \mid 0), \quad M(z) = M(0, z \mid 0) \quad \text{and} \quad \rho(x, z) = \rho(\delta_x, z \mid 0)
\]
for all $x \in \X$, where we note \eqref{eq:psi-x-D-x-phi}. 
Likewise, we consistently set
\[
	Z_\Lambda(\mathbf m, z \mid \mathbf b) : = Z(\mathbf m, z \1_\Lambda \mid \mathbf b), \quad M_\Lambda (\mathbf m, z \mid \mathbf b) : = M(\mathbf m, z \1_\Lambda \mid \mathbf b)
\]
and $\rho_\Lambda(\mathbf m, z \mid \mathbf b) : = \rho(\mathbf m, z \1_\Lambda \mid \mathbf b)$ for every $\Lambda \subset \X$.

Building on the relations sketched before Theorem~\ref{thm:dobrushin-sokal}, we extend a given $\Lambda \subset \X$ by some $x \in \X \setminus \Lambda$ and express the changes in partition function and Mayer series in terms of the effective activity $\widehat z_\Lambda(x, z) := z(x) \rho_\Lambda(x, z)$. 
Boundary conditions are straightforward to include, hence omitted here. 
We use \eqref{eq:Pi-z-semigroup} to formally equate $Z_{\Lambda \cup \{x\}}(z) = \pi^{z \1_{\{x\}} + z \1_\Lambda}[\kappa]$ with
\[
	\pi^{z \1_{\{x\}}}[\Pi^{z \1_\Lambda}[\kappa]] = \sum_{m \in \N_0} \frac{z(x)^m}{m!} Z_\Lambda(m \delta_x, z) = Z_\Lambda(z) + z(x) Z_\Lambda(x, z),
\]
where the last equality is due to $V$'s simplicity. 
If, for the last term, we then insert $Z_\Lambda(x, z) = \pi^{z \1_\Lambda}[\kappa \ast \psi(x, \bullet)] = Z_\Lambda(z) \rho_\Lambda(x, z)$, see \eqref{eq:pi-z-g-ast-h}, to arrive at
\be \label{eq:Z-induction-zhat}
	Z_{\Lambda \cup \{x\}}(z) = Z_\Lambda(z) (1 + \widehat z_\Lambda(x, z)).
\ee
To derive the logarithmic variant of the latter, $M_{\Lambda \cup \{x\}}(z) = \pi^z[\varphi \1_{\mathbf N(\Lambda \cup \{x\})}]$ is formally decomposed into the sum of $\pi^z[\varphi \1_{\mathbf N(\Lambda)}] = M_\Lambda(z)$ and
\[
	\pi^z[\varphi \1_{\mathbf N(\Lambda \cup \{x\}) \setminus \mathbf N(\Lambda)}] = \sum_{m \in \N} \frac{z(x)^m}{m!} M_\Lambda(m \delta_x, z),
\]
more or less using \eqref{eq:Pi-z-semigroup} again. 
Simplicity of $V$, through \eqref{eq:D-phi-hard-core-clique} and \eqref{eq:pi-z-conv-g}, gives
\[
	M_\Lambda(m \delta_x, z) = \pi^{z \1_\Lambda}[D^{m \delta_x} \varphi] = (-1)^{m-1} (m-1)! \rho_\Lambda(x, z)^m
\]
for all $m \in \N$. 
In total, we obtain the formal identity
\be \label{eq:M-induction-zhat}
	\pi^z[\varphi \1_{\mathbf N(\Lambda \cup \{x\}) \setminus \mathbf N(\Lambda)}] = \sum_{m \in \N} \frac{(-1)^{m-1} \widehat z_\Lambda(x, z)^m}{m},
\ee
the latter sum representing the logarithm of the factor $1 + \widehat z_\Lambda(x, z)$ in \eqref{eq:Z-induction-zhat}.

\subsection{Controlling correlations and the Kirkwood--Salsburg hierarchy}

Given $\mathbf m, \mathbf n, \mathbf b \in \mathbf N$, applying \eqref{eq:pi-z-g-ast-h} to \eqref{eq:psi-conditional-convolution} formally yields
\be \label{eq:D-rho}
	\rho(\mathbf m + \mathbf n, z \mid \mathbf b) = \rho(\mathbf m, z \mid \mathbf n + \mathbf b) \rho(\mathbf n, z \mid \mathbf b),
\ee
in addition to the unconditional identity $\rho(0, z \mid \mathbf b) = 1$. 
By comparison with \eqref{eq:D-kappa}, we can view the correlation $\rho(\mathbf m, z \mid \mathbf b)$ as an expansion around the Boltzmann factor $\kappa(\mathbf m \mid \mathbf b) = \psi(\mathbf m, 0 \mid \mathbf b) = \rho(\mathbf m, 0 \mid \mathbf b)$ as such.

We combine the latter insight with the \emph{Kirkwood--Salsburg hierarchy}, deduced as follows. 
The binomial formula
\[
	1_\ast = \1_{\{0\}} : \mathbf n  \mapsto (1 - 1)^\mathbf n = \sum_{\mathbf k \in \mathbf N} \binom{\mathbf n}{\mathbf k} (-1)^{\mathbf n - \mathbf k} = (1 \ast (-1)^\bullet)(\mathbf n),
\]
yields $1^{\ast -1} = (-1)^\bullet = (-1)^{|\bullet|}$ as the $\ast$-reciprocal of the constant function $1 = \1_\mathbf N$. 
Another trivariate function $\gamma : \mathbf N \times \mathbf N \times \mathbf N \to \C$ is then given by
\[
	\gamma(\mathbf m, \bullet \mid \mathbf b) : = \kappa(\mathbf m \mid \bullet + \mathbf b) \ast 1^{\ast -1} = \kappa(\mathbf m \mid \bullet + \mathbf b) \ast (-1)^{|\bullet|}
\]
for all $\mathbf m, \mathbf b \in \mathbf N$. 
Given $\mathbf m, \mathbf n, \mathbf b \in \mathbf N$, we employ the latter definition together with \eqref{eq:D-kappa} to equate $D^{\mathbf m + \mathbf n} \kappa(\bullet \mid \mathbf b)$ with
\[
	\kappa(\mathbf m \mid \mathbf n + \bullet + \mathbf b) D^\mathbf n \kappa(\bullet \mid \mathbf b) = (\gamma(\mathbf m, \bullet \mid \mathbf n + \mathbf b) \ast 1) D^\mathbf n \kappa(\bullet \mid \mathbf b),
\]
which, upon convolution with $\kappa(\bullet \mid \mathbf b)^{\ast -1}$ and rearranging, gives
\be \label{eq:psi-gamma-psi}
	\psi(\mathbf m + \mathbf n, \bullet \mid \mathbf b) = \sum_{\mathbf k, \mathbf k' \in \mathbf N} \binom{\bullet}{(\mathbf k, \mathbf k')} \gamma(\mathbf m, \mathbf k \mid \mathbf n + \mathbf b) \psi(\mathbf n + \mathbf k, \mathbf k' \mid \mathbf b).
\ee
If we apply \eqref{eq:pi-z-conv-g-generalised} to the latter, the resulting formal identity reads
\be \label{eq:KS}
	\rho(\mathbf m + \mathbf n, z \mid \mathbf b) = \pi^z[\gamma(\mathbf m, \bullet \mid \mathbf n + \mathbf b) \rho(\mathbf n + \bullet, z \mid \mathbf b)].
\ee
Assuming $|\mathbf m| = 1$, this is the \emph{Kirkwood--Salsburg equation}. 
Its classical use in a fixed point characterisation of Gibbs correlations can be found in, e.g., \cite{ref:ruelle, ref:gallavotti-miracle-sole, ref:jansen-cluster}. 
The usual Picard iteration can also be framed as a device to inductively control the dominating absolute series as in, e.g., \cite{ref:faris, ref:bissacot-fernandez-procacci, ref:jansen-cluster, ref:jansen-kolesnikov}.

Our final, rather specific, auxiliary relation regarding correlations complements \eqref{eq:Z-induction-zhat} and \eqref{eq:M-induction-zhat}. 
Accordingly, let $\Lambda \subset \X$ and $x \in \X \setminus \Lambda$. 
Then \eqref{eq:psi-x-D-x-phi} and \eqref{eq:Pi-z-semigroup} formally yield
\[
	z(x) \rho_{\Lambda \cup \{x\}}(x, z) = z(x) M_{\Lambda \cup \{x\}}(x, z) = \sum_{m \in \N_0} \frac{z(x)^{m+1}}{m!} M_\Lambda((m+1) \delta_x, z)
\]
and, if we proceed as in the derivation of \eqref{eq:M-induction-zhat}, we now obtain
\be \label{eq:z-rho-induction-zhat}
	z(x) \rho_{\Lambda \cup \{x\}}(x, z) = \widehat z_\Lambda(x, z) \sum_{m \in \N_0} (- \widehat z_\Lambda(x, z))^m.
\ee

\subsection{Conditional hypergraphs}

For the diagrammatic viewpoint, we use the conditional Mayer function $f : = \e^{-V} - 1 : \mathbf N \times \mathbf N \to \C$ to write
\[
	\kappa(\mathbf x \mid \mathbf y) = \sum_{\mathfrak h \in \mathfrak H_J} \prod_{e \in \mathfrak h} f(\mathbf x_e \mid \mathbf y) = \sum_{\mathfrak h \in \mathfrak H_{J \mid K}} \prod_{e \in \mathfrak h} f(\mathbf x_{e \cap J} + \mathbf y_{e \cap K})
\]
for arbitrary finitely indexed tuples $\mathbf x \in \X^J$, $\mathbf y \in \X^K$. 
Given implicitly disjoint sets $\U$ and $\V$, the conditional hypergraph notation means
\[
	\mathfrak H_{\U \mid \V} : = \{\mathfrak h \subset \{e \Subset \U \sqcup \V \mid e \cap \U \neq \varnothing\}\},
\]
i.e., $\V$ consists of purely auxiliary vertices unable to support edges alone. 
We do not use the diagrammatic representations of $\varphi$ or $\psi$, cf.\ \cite{ref:jansen-kolesnikov}, but that of $\gamma$ is essential to our argument. 
Let $x \in \X$ in addition to the finitely indexed tuples $\mathbf x \in \X^J$, $\mathbf y \in \X^K$. 
Then $\kappa(x \mid \mathbf x + \mathbf y)$, on the one hand, equals
\[
	\sum_{\mathfrak h \in \mathfrak H_{\{0\} \mid J}} \prod_{e \in \mathfrak h} f(x + \mathbf x_{e \setminus \{0\}} \mid \mathbf y) = \kappa(x \mid \mathbf y) \sum_{\mathfrak h \in \mathfrak H_J} \prod_{e \in \mathfrak h} f(x + \mathbf x_e \mid \mathbf y)
\]
and, by definition of $\gamma$, it also coincides with $\sum_{S \subset J} \gamma(x, \mathbf x_S \mid \mathbf y)$. 
A diagrammatic solution is provided by tracking the set $\mathrm{span}(\mathfrak h) : = \bigcup_{e \in \mathfrak h} e$ for an arbitrary hypergraph $\mathfrak h$ since, for a given vertex set $\V$, we clearly have
\[
	\mathfrak H_\V = \bigsqcup_{S \subset \V} \mathfrak S_S \quad \text{with} \quad \mathfrak S_\V : = \{\mathfrak h \in \mathfrak H_\V \mid \mathrm{span}(\mathfrak h) = \V\}.
\]
Comparison with the above representations of $\kappa(x \mid \mathbf x + \mathbf y)$ leads to
\[
	\gamma(x, \mathbf x \mid \mathbf y) = \kappa(x \mid \mathbf y) \sum_{\mathfrak h \in \mathfrak S_J} \prod_{e \in \mathfrak h} f(x + \mathbf x_e \mid \mathbf y).
\]
The latter expression can also be expanded further into
\be \label{eq:gamma-kappa-span-sum}
	\gamma(x, \mathbf x \mid \mathbf y) = \kappa(x \mid \mathbf y) \sum_{\mathfrak h \in \mathfrak S_{J \mid K}} \prod_{e \in \mathfrak h} f(x + \mathbf x_{e \cap J} + \mathbf y_{e \cap K}),
\ee
with $\mathfrak S_{\U \mid \V} : = \{\mathfrak h \in \mathfrak H_{\U \mid \V} \mid \mathrm{span}(\mathfrak h) \supset \U\}$ for arbitrary sets $\U$ and $\V$. 
If $V = V \1_{\{|\bullet| \leq 2\}}$ or equivalently $f = \e^{-V} - 1 = f \1_{\{|\bullet| \leq 2\}}$, the latter reads
\[
	\gamma(x, \mathbf x \mid \mathbf y) = \kappa(x \mid \mathbf y) \prod_{j \in J} f(x + x_j)
\]
and is straightforward to bound, cf.\ \cite{ref:ruelle, ref:ueltschi, ref:poghosyan-ueltschi, ref:fernandez-procacci, ref:bissacot-fernandez-procacci, ref:faris, ref:jansen-cluster, ref:jansen-kolesnikov} to name a few. 
Bounding the $2^{|J|}$ summands in $\gamma(x, \mathbf x \mid \mathbf y) = \sum_{S \subset J} (-1)^{|J \setminus S|} \kappa(x \mid \mathbf x_S + \mathbf y)$ individually is also an option, see \cite{ref:moraal}, as are hybrid approaches, treating (parts of) the two-particle component of $V$ separately from the rest as in \cite{ref:skrypnyk}, which is related to discretisations of continuum systems in cases where, e.g., $V \1_{\{|\bullet| = 2\}}$ is partially hard-core, cf.\ \cite{ref:procacci-scoppola-continuum}. 
Our chosen method of bounding $\gamma(x, \mathbf x \mid \mathbf y)$ combines \eqref{eq:gamma-kappa-span-sum} with the subsequent partition scheme for spanning hypergraphs, refining the approach in, e.g., \cite{ref:gallavotti-miracle-sole, ref:ruelle}. 

\begin{lemma} \label{lem:spanning-hypergraph-partition-scheme}
	Let $\U, \V$ be finite sets. 
	There exist maps $\mathfrak m : \mathfrak S_{\U \mid \V} \to \mathfrak S_{\U \mid \V}$ and $\mathfrak M : \mathfrak m(\mathfrak S_{\U \mid \V}) \to \mathfrak S_{\U \mid \V}$ such that, for all $\mathfrak h \in \mathfrak m(\mathfrak S_{\U \mid \V})$, we have
	\[
		\{\mathfrak h' \in \mathfrak S_{\U \mid \V} \mid \mathfrak m(\mathfrak h') = \mathfrak h\} = \{\mathfrak h' \in \mathfrak S_{\U \mid \V} \mid \mathfrak h \subset \mathfrak h' \subset \mathfrak M(\mathfrak h)\}.
	\]
	In particular, such maps yield
	\[
		\mathfrak S_{\U \mid \V} = \bigsqcup_{\mathfrak h \in \mathfrak m(\mathfrak S_{\U \mid \V})} \{\mathfrak h' \in \mathfrak S_{\U \mid \V} \mid \mathfrak h \subset \mathfrak h' \subset \mathfrak M(\mathfrak h)\}.
	\]
\end{lemma}

\begin{proof}
	Fix an arbitrary total order $\preceq$ on $\{e \Subset \U \sqcup \V \mid e \cap \U \neq \varnothing\}$, denoting $\prec$ the version that excludes equality. 
	For each $\mathfrak h \in \mathfrak S_{\U \mid \V}$, we define
	\[
		\mathfrak m(\mathfrak h) : = \{e \in \mathfrak h \mid e \cap \U \not\subset \mathrm{span}(\{e' \in \mathfrak h \mid e' \prec e\})\}.
	\]
	and it is easy to see that $\mathfrak m(\mathfrak h)$ can be assembled iteratively by considering the edges of $\mathfrak h$ in ascending order, that
	\[
		\mathrm{span}(\{e' \in \mathfrak h \mid e' \prec e\}) \cap \U = \mathrm{span}(\{e' \in \mathfrak m(\mathfrak h) \mid e' \prec e\}) \cap \U
	\]
	for all $e \in \mathfrak h$ and that indeed $\mathrm{span}(\mathfrak m(\mathfrak h)) \cap \U = \mathrm{span}(\mathfrak h) \cap \U = \U$. 
	We thereby obtain a map $\mathfrak m : \mathfrak S_{\U \mid \V} \to \mathfrak S_{\U \mid \V}$ with the image
	\[
		\mathfrak m(\mathfrak S_{\U \mid \V}) = \{\mathfrak h \in \mathfrak S_{\U \mid \V} \mid \forall e \in \mathfrak h : e \cap \U \not\subset \mathrm{span}(\{e' \in \mathfrak h \mid e' \prec e\})\}
	\]
	and, for all $\mathfrak h \in \mathfrak m(\mathfrak S_{\U \mid \V})$ and $\mathfrak h' \in \mathfrak S_{\U \mid \V}$, we have $\mathfrak m(\mathfrak h') = \mathfrak h$ if and only if
	\[
		\mathfrak h \subset \mathfrak h' \subset \mathfrak h \sqcup \{e \Subset \U \sqcup \V \mid \varnothing \neq e \cap \U \subset \mathrm{span}(\{e' \in \mathfrak h \mid e' \prec e\})\} = : \mathfrak M(\mathfrak h).
	\]
	This completes the proof.
\end{proof}

Lemma~\ref{lem:spanning-hypergraph-partition-scheme} is used in the second part of our proof of Theorem~\ref{thm:dobrushin-sokal} below. 
There, we briefly comment on the utility of this conditional version of the analogous argument in Part~I.

\section{Proof of our main results} \label{sec:proof-main}

For this section, we fix $z : \X \to \C$ and use the shorthand
\[
	\xi_\Lambda(\mathbf m \mid \mathbf b) : = |z|^\mathbf m \pi^{|z| \1_\Lambda}[|\psi(\mathbf m, \bullet \mid \mathbf b)|] \in \overline{\R}_+
\]
for all $\Lambda \subset \X$ and $\mathbf m, \mathbf b \in \mathbf N$. 
Theorem~\ref{thm:dobrushin-sokal} is proved by establishing that its hypotheses imply that, for all $\Lambda \subset \X$, $\mathbf m \in \mathbf N$ and $B \Subset \X$, we have
\be \label{eq:xi-leq-alpha-r-power-simple}
	\xi_\Lambda(\mathbf m \mid B) \leq (\alpha \1_\Lambda + r \1_{\X \setminus \Lambda})^\mathbf m \1_{\mathbf F_{\mid B}}(\mathbf m)
\ee
with $\alpha = \frac{r}{1-r} : \X \to \R_+$ and $r = \frac{\alpha}{1+\alpha} : \X \to [0, 1)$. 
Let us stress again that the simplicity of $V$ is essential to our argument and that it immediately validates the latter bound for $\mathbf m \notin \mathbf F_{\mid B} = \{X \Subset \X \setminus B\}$. 
Otherwise, we proceed inductively based on \eqref{eq:D-rho}, \eqref{eq:KS} and \eqref{eq:z-rho-induction-zhat}. 
Lemma~\ref{lem:spanning-hypergraph-partition-scheme} is used to connect the induction step to the hypothesis of Theorem~\ref{thm:dobrushin-sokal}.

Part~I features two proofs of analogous correlation bounds but not for the cluster expansions: 
the primary proof there adapts the argument of \cite{ref:bencs-buys}, leading to an induction along $\Lambda$; 
the alternative one ironically strips the Kirkwood--Salsburg-based ansatz in \cite{ref:jansen-kolesnikov} of its cluster expansion context, using a fixed $\Lambda$. 
The subsequent argument combines the two approaches for a streamlined treatment of Theorem~\ref{thm:dobrushin-sokal} and Corollary~\ref{cor:Z-M-r-bounds}.

\subsection{Proof of Theorem~\ref{thm:dobrushin-sokal} - setup of the induction}

Fix $\Lambda \subset \X$. 
For all $\mathbf m, \mathbf n \in \mathbf N$, applying Lemma~\ref{lem:pi-z-conv-g-generalised} to \eqref{eq:psi-conditional-convolution} yields
\[
	\xi_\Lambda(\mathbf m + \mathbf n \mid \mathbf b) \leq \xi_\Lambda(\mathbf m \mid \mathbf n + \mathbf b) \xi_\Lambda(\mathbf n \mid \mathbf b)
\]
with $\xi_\Lambda(\varnothing \mid \mathbf b) = \pi^{|z| \1_\Lambda}[1_\ast] = 1$. 
In particular, we use the variant
\be \label{eq:xi-leq-prod-xi}
	\xi_\Lambda(X \mid B) \leq \prod_{x \in X} \xi_\Lambda(x \mid \{y \in X \mid y \prec x\} \cup B)
\ee
for all $X, B \Subset \X$ with $X \cap B = \varnothing$, where $\prec$ denotes an arbitrary total order on $X$ with equality excluded.

In view of the latter, we now also fix $B \Subset \X$ and $x \in \X \setminus B$. 
Applying Lemmas~\ref{lem:superposition-principle} and \ref{lem:pi-z-conv-g-generalised} along the derivation of \eqref{eq:z-rho-induction-zhat} nets us
\be \label{eq:xi-x-in-x-out}
	\xi_{\Lambda \cup \{x\}}(x \mid B) \leq \xi_{\Lambda \setminus \{x\}}(x \mid B) \sum_{m \in \N_0} \xi_{\Lambda \setminus \{x\}}(x \mid B)^m.
\ee
This upper bound is finite if and only if $\xi_{\Lambda \setminus \{x\}}(x \mid B) < 1$, in which case it is equal to $\frac{\xi_{\Lambda \setminus \{x\}}(x \mid B)}{1-\xi_{\Lambda \setminus \{x\}}(x \mid B)}$. 
In particular, $\xi_{\Lambda \cup \{x\}}(x \mid B) \leq \alpha(x) = \frac{r(x)}{1-r(x)}$ follows whenever $\xi_{\Lambda \setminus \{x\}}(x \mid B) \leq r(x) < 1$.

For the induction step, we use the Kirkwood--Salsburg hierarchy in the form
\be \label{eq:xi-leq-sum-gamma-xi}
	\xi_\Lambda(x \mid B) \leq |z(x)| \sum_{Y \Subset \Lambda \setminus B} |\gamma(x, Y \mid B)| \xi_\Lambda(Y \mid B).
\ee
This is obtained by applying Lemma~\ref{lem:pi-z-conv-g-generalised} to \eqref{eq:psi-gamma-psi}, thereby bounding the left-hand side of \eqref{eq:xi-leq-sum-gamma-xi} by $|z(x)| \pi^{\1_\Lambda}[|\gamma(x, \bullet \mid B)| \xi_\Lambda(\bullet \mid B)]$, which becomes the right-hand side upon accounting for $V$'s simplicity.

The combination of the above inequalities gives the following.

\begin{lemma} \label{lem:xi-KS-induction}
	Suppose that, for all $x \in \X$ and $B \Subset \X$ with $x \notin B$,
	\[
		|z(x)| \sum_{Y \Subset \X \setminus (\{x\} \cup B)} |\gamma(x, Y \mid B)| \alpha^Y \leq r(x).
	\]
	Then \eqref{eq:xi-leq-alpha-r-power-simple} holds for all $\Lambda \subset \X$, $\mathbf m \in \mathbf N$ and $B \Subset \X$.
\end{lemma}

\begin{proof}
	By monotone convergence, it suffices to show that the conclusion holds for all $\Lambda \Subset \X$ and we do so by induction.
	
	For $\Lambda = \varnothing$ and $B \Subset \X$, \eqref{eq:xi-leq-prod-xi} implies \eqref{eq:xi-leq-alpha-r-power-simple} with arbitrary $\mathbf m \in \mathbf N$ as soon as we validate the case $\mathbf m = \delta_x$ with arbitrary $x \in \X \setminus B$. 
	Indeed,
	\[
		\xi_\varnothing(x \mid B) = |z(x)| |\psi(x, \varnothing \mid B)| = |z(x)| |\gamma(x, \varnothing \mid B)| \leq r(x)
	\]
	follows from $\psi(x, \varnothing \mid B) = \kappa(x \mid B) = \gamma(x, \varnothing \mid B)$ and the hypothesis.
	
	Now fix some $\Lambda \Subset \X$ and $B \Subset \X$. 
	Proving \eqref{eq:xi-leq-alpha-r-power-simple} again reduces to the case $\mathbf m = \delta_x$ with $x \in \X \setminus B$, i.e., it remains to show
	\[
		\xi_\Lambda(x \mid B) \leq \alpha(x) \1_\Lambda(x) + r(x) \1_{\X \setminus \Lambda}(x) = \frac{r(x)}{1 - r(x) \1_\Lambda(x)}.
	\]
	If $x \in \Lambda$, the latter inductively follows from $\xi_{\Lambda \setminus \{x\}}(x \mid B) \leq r(x) < 1$ via \eqref{eq:xi-x-in-x-out}. 
	If, on the other hand, $x \notin \Lambda$, then we obtain
	\[
		\xi_\Lambda(x \mid B) \leq |z(x)| \sum_{Y \Subset \Lambda \setminus B} |\gamma(x, Y \mid B)| \alpha^Y \leq r(x),
	\]
	where the first inequality follows from a combination of \eqref{eq:xi-leq-sum-gamma-xi}, \eqref{eq:xi-leq-prod-xi} and the previous case while the second inequality is just the lemma's hypothesis.
\end{proof}

\subsection{Proof of Theorem~\ref{thm:dobrushin-sokal} - derivation of the criterion}

We now utilise Lemma~\ref{lem:spanning-hypergraph-partition-scheme} to insert an intermediate bound into the the hypothesis of Lemma~\ref{lem:xi-KS-induction}. 
Given $x \in \X$ and $B, Y \Subset \X \setminus \{x\}$ with $Y \cap B = \varnothing$,  \eqref{eq:gamma-kappa-span-sum} yields
\[
	\gamma(x, Y \mid B) = \kappa(x \mid B) \sum_{\mathfrak h \in \mathfrak S_{Y \mid B}} \prod_{e \in \mathfrak h} f(x + e).
\]
We bound the sum over $\mathfrak S_{Y \mid B}$ by first expanding it via
\[
	\sum_{\mathfrak h \in \mathfrak S_{Y \mid B}} \prod_{e \in \mathfrak h} f(x + e) = \sum_{\mathfrak h \in \mathfrak m(\mathfrak S_{Y \mid B})} \prod_{e \in \mathfrak h} f(x + e) \prod_{e \in \mathfrak M(\mathfrak h) \setminus \mathfrak h} (1 + f(x + e)),
\]
where $\mathfrak m : \mathfrak S_{Y \mid B} \to \mathfrak S_{Y \mid B}$ and $\mathfrak M : \mathfrak m(\mathfrak S_{Y \mid B}) \to \mathfrak S_{Y \mid B}$ are the functions from the proof of Lemma~\ref{lem:spanning-hypergraph-partition-scheme} with $\U = Y$ and $\V = B$. 
Writing $1+f = \e^{-V}$ and multiplying by $\kappa(x \mid B) = \e^{- V(x \mid B)} = \prod_{e \subset B} \e^{- V(x + e)}$ gives
\[
	\gamma(x, Y \mid B) = \prod_{e \subset B} \e^{- V(x + e)} \sum_{\mathfrak h \in \mathfrak m(\mathfrak S_{Y \mid B})} \prod_{e \in \mathfrak h} f(x + e) \prod_{e \in \mathfrak M(\mathfrak h) \setminus \mathfrak h} \e^{-V(x + e)}.
\]
If we then take absolute values and multiply by $\alpha^Y$, we get
\begin{align*}
	& |\gamma(x, Y \mid B)| \alpha^Y \\
	& \quad \leq \prod_{e \subset B} |\e^{- V(x + e)}| \sum_{\mathfrak h \in \mathfrak m(\mathfrak S_{Y \mid B})} \prod_{e \in \mathfrak h} |f(x + e)| \alpha^{Y \cap e \setminus \mathrm{span}(\mathfrak h_{\prec e})} \prod_{e \in \mathfrak M(\mathfrak h) \setminus \mathfrak h} |\e^{-V(x + e)}|
\end{align*}
with $\mathfrak h_{\prec e} = \{e \in \mathfrak h \mid e' \prec e\}$ and 
$\alpha^{Y \cap e \setminus \mathrm{span}(\mathfrak h_{\prec e})} \leq \max\{\alpha^S \mid \varnothing \neq S \subset e \setminus B\}$ for each $\mathfrak h \in \mathfrak m(\mathfrak S_{Y \mid B})$ and $e \in \mathfrak h$. 
For every $e \Subset \X \setminus \{x\}$, we then use the shorthand
\[
	\mathrm{m}(x, e \mid B) : = \max\{|\e^{-V(x+e)}|, 1 + |f(x + e)| \alpha^S \mid \varnothing \neq S \subset e \setminus B\}
\]
and note that $\mathrm{m}(x, e \mid B) \geq 1$ whenever $e \setminus B \neq \varnothing$ so the above implies
\begin{align*}
	& |\gamma(x, Y \mid B)| \alpha^Y \\
	& \quad \leq \prod_{e \subset B} \mathrm{m}(x, e \mid B) \sum_{\mathfrak h \in \mathfrak m(\mathfrak S_{Y \mid B})} \prod_{e \in \mathfrak h} (\mathrm{m}(x, e \mid B)) - 1) \prod_{e \in \mathfrak M(\mathfrak h)} \mathrm{m}(x, e \mid B) \\
	& \quad = \prod_{e \subset B} \mathrm{m}(x, e \mid B) \sum_{\mathfrak h \in \mathfrak S_{Y \mid B}} \prod_{e \in \mathfrak h} (\mathrm{m}(x, e \mid B)) - 1).
\end{align*}
The last identity is a reverse application of Lemma~\ref{lem:spanning-hypergraph-partition-scheme} and summing over the possible choices of $Y$ nets us
\begin{align}
	& \sum_{Y \Subset \X \setminus (\{x\} \cup B)} |\gamma(x, Y \mid B)| \alpha^Y \notag \\
	& \quad \leq \prod_{e \subset B} \mathrm{m}(x, e \mid B) \sum_{Y \Subset \X \setminus (\{x\} \cup B)} \sum_{\mathfrak h \in \mathfrak S_{Y \mid B}} \prod_{e \in \mathfrak h} (\mathrm{m}(x, e \mid B)) - 1) \notag \\
	& \quad = \prod_{e \Subset \X \setminus \{x\}} \mathrm{m}(x, e \mid B). \label{eq:sum-gamma-alpha-leq-Gamma-alpha}
\end{align}
The last equality is based on equating $\bigsqcup_{Y \Subset \X \setminus (\{x\} \cup B)} \mathfrak S_{Y \mid B}$ with
\[
	 \{\mathfrak h \in \mathfrak H_{\X \setminus (\{x\} \cup B) \mid B} \mid |\mathrm{span}(\mathfrak h)| < \infty\} = \{\mathfrak h \in \mathfrak H_{\X \setminus (\{x\} \cup B) \mid B} \mid |\mathfrak h| < \infty\},
\]
noting again that $\mathrm{m}(x, e \mid B) \geq 1$ for all $e \Subset \X \setminus \{x\}$ with $e \setminus B \neq \varnothing$.

At this point, it is evident that all the hypotheses of Lemma~\ref{lem:xi-KS-induction} are simultaneously satisfied whenever $|z(x)| \prod_{e \Subset \X \setminus \{x\}} \mathrm{m}(x, e | \varnothing) \leq r(x)$, i.e.,
\[
	|z(x)| \prod_{e \Subset \X \setminus \{x\}} \max\{|\e^{-V(x+e)}|, 1 + |f(x + e)| \alpha^S \mid \varnothing \neq S \subset e\} \leq r(x),
\]
for all $x \in \X$, which is precisely the hypothesis \eqref{eq:dobrushin-sokal} of Theorem~\ref{thm:dobrushin-sokal}, whose conclusions are, on the other hand, subsumed by those of Lemma~\ref{lem:xi-KS-induction}.
\hfill $\square$
\medskip

The aforementioned alternative proof in Part~I basically establishes \eqref{eq:xi-leq-alpha-r-power-simple} via an unconditional version of Lemma~\ref{lem:spanning-hypergraph-partition-scheme}, i.e., with $\V = \varnothing$. 
The accompanying proof of why the condition \eqref{eq:dobrushin-sokal} is ``stable under conditioning'' is made obsolete here by employing the conditional version from the start.

\subsection{Proof of Corollary~\ref{cor:Z-M-r-bounds}}

The overarching assumption in this proof is \eqref{eq:zhat-leq-r} in the form
\be \label{eq:zhat-xi-r-1}
	|\widehat z_{\Lambda \setminus \{x\}}(x, z)| \leq \xi_{\Lambda \setminus \{x\}}(x) \leq r(x) < 1,
\ee
for all $\Lambda \subset \X$ and $x \in \X$. 
Simple boundary conditions are easily included but omitted for better readability, just like in the derivations of \eqref{eq:Z-induction-zhat} and \eqref{eq:M-induction-zhat}, which the following arguments are based upon.

Suppose that $\Lambda \subset \X$ and $x \in \X \setminus \Lambda$. 
We reiterate the derivation of \eqref{eq:Z-induction-zhat} but with absolute values inserted everywhere.  Lemmas~\ref{lem:superposition-principle} and \ref{lem:pi-z-conv-g-generalised} thereby yield
\[
	\pi^{|z| \1_{\{x\} \cup \Lambda}}[|\kappa|] \leq \pi^{|z| \1_\Lambda}[|\kappa|] (1 + \xi_\Lambda(x))
\]
and, starting from $\pi^{|z| \1_\varnothing} = |\kappa(0)| = 1 = (1 + r)^\varnothing$, we can inductively apply \eqref{eq:zhat-xi-r-1} to obtain $\pi^{|z| \1_\Lambda}[|\kappa|] \leq (1 + r)^\Lambda$ whenever $\Lambda$ is finite and, for arbitrary $\Lambda \subset \X$, the latter bound then also follows by monotone convergence. 
In particular, if $|Z_\Lambda(z)| \leq \pi^{|z| \1_\Lambda}[|\kappa|] < \infty$, the absolute values can be dropped to obtain
\[
	Z_{\{x\} \cup \Lambda}(z) = Z_\Lambda(z) (1 + \widehat z_\Lambda(x, z)) \in \C
\]
as well as $|Z_\Lambda(z)| (1 - r(x)) \leq |Z_{\{x\} \cup \Lambda}(z)| \leq |Z_\Lambda(z)| (1 + r(x))$. 
Accordingly,
\[
	0 < (1 - r)^\Lambda \leq |Z_\Lambda(z)| \leq \pi^{|z| \1_\Lambda}[|\kappa|] \leq (1 + r)^\Lambda < \infty
\]
follows inductively if $\Lambda \Subset \X$ and combining monotone and dominated convergence extends the result to the case of merely $r$-finite $\Lambda \subset \X$.

The mostly analogous argument for our Mayer series imitates the derivation of \eqref{eq:M-induction-zhat}, where Lemmas~\ref{lem:superposition-principle} and \ref{lem:pi-z-conv-g-generalised} give
\[
	\pi^{|z| \1_{\{x\} \cup \Lambda}}[|\varphi|] \leq \pi^{|z| \1_\Lambda}[|\varphi|] + \pi^{|z|}[|\varphi| \1_{\mathbf N(\{x\} \cup \Lambda) \setminus \mathbf N(\Lambda)}]
\]
with
\[
	\pi^{|z|}[|\varphi| \1_{\mathbf N(\{x\} \cup \Lambda) \setminus \mathbf N(\Lambda)}] \leq \sum_{m \in \N} \frac{\xi_\Lambda(x)^m}{m} \leq \sum_{m \in \N} \frac{r(x)^m}{m} = - \log(1 - r(x)) < \infty,
\]
cf.\ \eqref{eq:zhat-xi-r-1}. 
In particular, this already justifies
\[
	\pi^{z}[\varphi \1_{\mathbf N(\{x\} \cup \Lambda) \setminus \mathbf N(\Lambda)}] = \sum_{m \in \N} \frac{(-1)^{m-1} \widehat z_\Lambda(x, z)^m}{m} = \log(1 + \widehat z(x, z)) \in \C.
\]
Sidestepping a proper induction argument, we now fix some total order $\preceq$ on $\Lambda$ and set $\Lambda_{\prec x} : = \{y \in \Lambda \mid y \prec x\}$ for all $x \in \Lambda$, where $\prec$ again denotes $\preceq$ without equality. 
Since clearly
\[
	\mathbf N(\Lambda) = \{0\} \sqcup \bigsqcup_{x \in \Lambda} (\mathbf N(\{x\} \cup \Lambda_{\prec x}) \setminus \mathbf N(\Lambda_{\prec x}))
\]
and $\pi^{|z| \1_\varnothing}[|\varphi|] = |\varphi(0)| = 0$, we have
\[
	\pi^{|z| \1_\Lambda}[|\varphi|] = \pi^{|z|}[|\varphi| \1_{\mathbf N(\Lambda)}] = \sum_{x \in \Lambda} \pi^{|z|}[|\varphi| \1_{\mathbf N(\{x\} \cup \Lambda_{\prec x}) \setminus \mathbf N(\Lambda_{\prec x})}]
\]
and bounding each summand according to the above yields
\[
	\pi^{|z| \1_\Lambda}[|\varphi|] \leq - \sum_{x \in \Lambda} \log(1 - r(x)) \leq - \log(1 - r)[\Lambda].
\]
Hence, if $\Lambda$ is $r$-finite, then $|M_\Lambda(z)| \leq \pi^{|z| \1_\Lambda}[|\varphi|] \leq - \log(1 - r)[\Lambda] < \infty$, so
\[
	|Z_\Lambda(z)| \leq \pi^{|z| \1_\Lambda}[|\kappa|] = \pi^{|z| \1_\Lambda}[|\exp_\ast \varphi|] \leq \e^{\pi^{|z| \1_\Lambda}(z)[|\varphi|]} 	< \infty
\]
and Lemmas~\ref{lem:pi-z-conv-g-generalised} and \ref{lem:superposition-principle} let us drop all absolute values to arrive at
\[
	Z_\Lambda(z) = \exp(M_\Lambda(z)) = \prod_{x \in \Lambda} (1 + \widehat z_{\Lambda_{\prec x}}(x, z)) \in \C \setminus \{0\},
\]
the latter identities also yielding the bounds
\[
	0 < (1 - r)^\Lambda \leq |Z_\Lambda(1)| \leq (1 + r)^\Lambda < \infty.
\]
With this, the proof of Corollary~\ref{cor:Z-M-r-bounds} is complete.
\hfill $\square$

\subsection{Proof of Theorem~\ref{thm:dobrushin-sokal-polymer}}

Given the long history of the polymer expansion, we keep this part of the exposition rather short. 
We assume $V(x) = 0$ for all $x \in \X$, disregarding the easy but not particularly insightful relaxation of this assumption, and work with the condition of Gruber and Kunz. 
They defined polymer systems slightly differently but their bounds throughout \cite[Subsection~4.2]{ref:gruber-kunz} are still recognisable in the following version adopted from \cite{ref:fernandez-procacci, ref:bissacot-fernandez-procacci, ref:jansen-kolesnikov}. 
$\alpha : \X \to \R_+$ and $r = \frac{\alpha}{1+\alpha} : \X \to [0, 1)$ are as before.

\begin{prop}[{\cite[Theorem~2.4]{ref:bissacot-fernandez-procacci}/\cite[Corollary~2.8]{ref:jansen-kolesnikov}}] \label{prop:gruber-kunz}
	Let $w : \P \to \C$ and suppose that
	\[
		\sum_{X \Subset \X : x \in X} |w(X)| (1 + \alpha)^X \leq \alpha(x)
	\]
	for all $x \in \X$. 
	Then, for every $r$-finite $\Lambda \subset \X$, one has
	\[
		0 < (1-r)^\Lambda \leq \Xi_\Lambda(-|w|) \leq |\Xi_\Lambda(w)| \leq \Xi_\Lambda(|w|) \leq (1+r)^\Lambda < \infty,
	\]
	\[
		|\Omega_\Lambda(w)| \leq \sum_{\mathbf n \in \mathbf N(P(\Lambda))} \frac{|w|^\mathbf n}{\mathbf n!} |\varphi_\P(\mathbf n)| = - \Omega_\Lambda(-|w|) \leq - \log (1 - r)[\Lambda] < \infty,
	\]
	and $\Xi_\Lambda(w) = \exp(\Omega_\Lambda(w)) \in \C \setminus \{0\}$.
\end{prop}

When looking up this result, one should mentally write $\alpha = \e^a - 1$ and note the alternating sign property $\varphi_\P = (-1)^{|\bullet|-1} |\varphi_\P|$, see, e.g., \cite{ref:scott-sokal, ref:fernandez-procacci, ref:bissacot-fernandez-procacci, ref:jansen-kolesnikov}. 
The extension from finite to $r$-finite reference volumes is again just monotone/dominated convergence. 
We connect the hypothesis of Theorem~\ref{thm:dobrushin-sokal-polymer} to that of Proposition~\ref{prop:gruber-kunz} via the following adaptation of the second part of our proof of Theorem~\ref{thm:dobrushin-sokal}.

\begin{lemma} \label{lem:dobrushin-sokal-subsets}
	Let $\zeta, A : \X \to \R_+$ and suppose that $\zeta \, \Gamma_A \leq A$. 
	Then, for all $x \in \X$, we also have
	\[
		\sum_{X \Subset \X : x \in X} \zeta^X |\varphi(X)| \leq \zeta(x) \Gamma_A(x) \leq A(x).
	\]
\end{lemma}

\begin{proof}
	If we set $A_0 := A$ and inductively define
	\[
		A_{n+1} : = \zeta \, \Gamma_{A_n} \leq A_n \quad \text{with} \quad \zeta \, \Gamma_{A_{n+1}} \leq \zeta \, \Gamma_{A_n} = A_{n+1}
	\]
	for all $n \in \N_0$, followed by $A_\infty : = \lim_{n\to\infty} A_n = \inf_{n \in \N_0} A_n$, then pointwise applied dominated convergence guarantees
	\[
		\zeta \, \Gamma_{A_\infty} = \lim_{n \to \infty} \zeta \, \Gamma_{A_n} = \lim_{n \to \infty} A_{n+1} = A_\infty.
	\]
	By replacing $A$ with $A_\infty$, we therefore assume $\zeta \, \Gamma_A = A$ without loss of generality. 
	At this point, \eqref{eq:psi-x-D-x-phi} implies that it is sufficient to show
	\[
		\zeta^X \sum_{Y \Subset \X \setminus X} \zeta^Y |\psi(X, Y)| = \zeta^X \pi^\zeta[|\psi(X, \bullet)| \1_{\mathbf F(\X \setminus X)}] \leq A^X
	\]
	for all $X \Subset \X$. 
	By a slight adaptation of Lemma~\ref{lem:xi-KS-induction}, in which the use of \eqref{eq:xi-x-in-x-out} becomes obsolete, this can be further reduced to showing
	\[
		\zeta(x) \sum_{Y \Subset \X \setminus (\{x\} \cup B)} |\gamma(x, S \mid B)| A^Y \leq \zeta(x) \Gamma_A(x) = A(x)
	\]
	for arbitrary $x \in \X$ but the inequality here is already established in \eqref{eq:sum-gamma-alpha-leq-Gamma-alpha}.
\end{proof}

Lemma~\ref{lem:dobrushin-sokal-subsets} is our current iteration of improving upon \cite[Theorem~2]{ref:procacci-scoppola} as it pertains to lattice gases in our sense. 
The initial simplification in our preceding proof is also heavily inspired by, e.g., \cite{ref:fernandez-procacci, ref:bissacot-fernandez-procacci}, \cite{ref:faris} or \cite{ref:jansen-hierarchical}.

\begin{proof}[Proof of Theorem~\ref{thm:dobrushin-sokal-polymer}]
	Let $z : \X \to \C \setminus \{-1\}$ and fix a choice of $A : \X \to \R_+$ satisfying \eqref{eq:dobrushin-sokal-polymer}. 
	After possibly decreasing $\alpha$, we assume $A\frac{\Gamma_A-1}{\Gamma_A} = \alpha$  without loss and accordingly
	\[
		|{\textstyle \frac{z}{1+z}}| (1+\alpha) \Gamma_A = |{\textstyle \frac{z}{1+z}}| (\Gamma_A + A (\Gamma_A-1)) = |{\textstyle \frac{z}{1+z}}| ((1+A) \Gamma_A - A) \leq A.
	\]
	For $\zeta = |\frac{z}{1+z}| (1 + \alpha)$ and an arbitrary $x \in \X$, Lemma~\ref{lem:dobrushin-sokal-subsets} now yields
	\[
		\sum_{X \Subset \X : x \in X} \zeta^X |\varphi(X) - \1_{\{|X|=1\}}| \leq \zeta(x) (\Gamma_A(x) - 1),
	\]
	noting that $\varphi(x) = 1$ due to $V(x) = f(x) = 0$. 
	Observe that
	\[
		\zeta \, (\Gamma_A - 1) = |{\textstyle \frac{z}{1+z}}| (1 + \alpha) (\Gamma_A - 1) = |{\textstyle \frac{z}{1+z}}| (1 + \alpha) \Gamma_A {\textstyle \frac{\Gamma_A - 1}{\Gamma_A}} \leq A {\textstyle \frac{\Gamma_A - 1}{\Gamma_A}} = \alpha
	\]
	so, writing $w = (\frac{z}{1+z})^\bullet (\varphi - \1_{\{|\bullet| = 1\}}) : \P \to \C$, we obtain
	\[
		\sum_{X \Subset \X : x \in X} |w(X)|(1+\alpha)^X = \sum_{X \Subset \X : x \in X} \zeta^X |\varphi(X) - \1_{\{|X|=1\}}| \leq \alpha(x)
	\]
	for all $x \in \X$. 
	An application of Proposition~\ref{prop:gruber-kunz} completes the proof.
\end{proof}

\subsubsection*{Statement on data availability and no conflict of interest}

Data sharing is not applicable. 
We do not analyse or generate any datasets because our work proceeds within a theoretical and mathematical approach. 

The author has no competing interests to declare that are relevant to the content of this article.

\subsubsection*{Acknowledgement.}

This research has partially been funded by the Deutsche Forschungsgemeinschaft (DFG) through grant SPP 2265 ``Random Geometric Systems'', Project P13, followed by support under Germany’s excellence strategy EXC-2111-390814868 until September 2025. 
The author would like to thank Sabine Jansen and Leonid Kolesnikov for helpful discussions.

\end{document}